\documentclass[acmsmall,screen]{acmart}

\setcopyright{rightsretained}
\acmPrice{}
\acmDOI{10.1145/3434292}
\acmYear{2021}
\copyrightyear{2021}
\acmSubmissionID{popl21main-p70-p}
\acmJournal{PACMPL}
\acmVolume{5}
\acmNumber{POPL}
\acmArticle{11}
\acmMonth{1}

\usepackage[ruled]{algorithm2e}
\usepackage{proof}
\usepackage{wrapfig}
\usepackage{listings}
\usepackage{stmaryrd}
\usepackage{color}
\usepackage{hyperref}
\usepackage[capitalise]{cleveref}
\usepackage{graphicx}
\usepackage[all]{xy}
\usepackage{tikz}
\usetikzlibrary{external}
\theoremstyle{plain}
\newtheorem{theorem}{Theorem}[section]
\newtheorem*{theorem*}{Theorem}
\newtheorem{corollary}[theorem]{Corollary}

\newtheorem{observation}[theorem]{Observation}
\newtheorem{lemma}[theorem]{Lemma}
\newtheorem{proposition}[theorem]{Proposition}
\theoremstyle{definition}
\newtheorem{definition}[theorem]{Definition}
\newtheorem{remark}[theorem]{Remark}
\newtheorem{example}[theorem]{Example}
\newtheorem{notation}[theorem]{Notation}

\newcommand*{\mlstinline}[1]{\text{\lstinline|#1|}}

\DeclareMathOperator{\id}{id}
\newcommand{\R}{\mathbb R}

\newcommand{\N}{\mathbb N}
\newcommand{\T}{\mathbb{T}}
\renewcommand{\P}{\mathrm{Pr}}
\newcommand{\true}{\mathtt{true}}
\newcommand{\false}{\mathtt{false}}
\newcommand{\new}{\mathtt{new}}
\newcommand{\cpy}{\mathrm{copy}}

\newcommand{\del}{\mathrm{del}}
\newcommand{\dom}{\mathrm{dom}}

\newcommand{\cat}{\mathbb}
\newcommand{\catname}{\mathbf}
\newcommand{\sbs}{\catname{Sbs}}
\newcommand{\qbs}{\catname{Qbs}}
\newcommand{\meas}{\catname{Meas}}

\newcommand{\set}{\catname{Set}}

\renewcommand{\exp}{\mathrm{exp}}
\newcommand{\Exp}{\mathrm{Exp}}
\newcommand{\val}{\mathrm{val}}
\newcommand{\Val}{\mathrm{Val}}

\newcommand{\sem}[1]{\llbracket{#1}\rrbracket}
\newcommand{\norm}[1]{\langle {#1} \rangle} 

\renewcommand{\oplus}{\uplus} 

\newcommand{\dd}[2]{d#1(#2)} 

\newcommand{\defeq}{\stackrel{\text{def}}=}

\newcommand{\tboolintro}{\mathsf{bool}}
\newcommand{\tnamesintro}{\mathsf{name}}
\newcommand{\trealintro}{\mathsf{real}}
\newcommand{\tnames}{\mathsf{N}}
\newcommand{\tbool}{\mathsf{B}}
\newcommand{\ite}[3]{\mathtt{if}~#1~\mathtt{then}~#2~\mathtt{else}~#3}
\newcommand{\letin}[3]{\mathtt{let}~#1 \leftarrow #2~\mathtt{in}~#3}

\usepackage{soul}

\DeclareMathOperator{\Priv}{Priv}
\DeclareMathOperator{\Leak}{Leak}
\DeclareMathOperator{\Safe}{Safe}

\begin{document}

\title{Probabilistic Programming Semantics for Name Generation}

\author{Marcin Sabok}
\affiliation{
	\department{Department of Mathematics and Statistics}
	\institution{McGill University}
	\city{Montreal}
	\country{Canada}
}
\email{marcin.sabok@mcgill.ca}

\author{Sam Staton}
\affiliation{
	\department{Department of Computer Science}
	\institution{University of Oxford}
	\city{Oxford}
	\country{United Kingdom}
}
\email{sam.staton@cs.ox.ac.uk}

\author{Dario Stein}
\affiliation{
  \department{Department of Computer Science}
  \institution{University of Oxford}
  \city{Oxford}
  \country{United Kingdom}
}
\email{dario.stein@cs.ox.ac.uk}

\author{Michael Wolman}
\affiliation{
	\department{Department of Mathematics and Statistics}
	\institution{McGill University}
	\city{Montreal}
	\country{Canada}
}
\email{michael.wolman@mail.mcgill.ca}

\begin{abstract}
We make a formal analogy between random sampling and fresh name generation. We show that quasi-Borel spaces, a model for probabilistic programming, can soundly interpret the $\nu$-calculus, a calculus for name generation. Moreover, we prove that this semantics is fully abstract up to first-order types. This is surprising for an `off-the-shelf' model, and requires a novel analysis of probability distributions on function spaces. Our tools are diverse and include descriptive set theory and normal forms for the $\nu$-calculus.

\end{abstract}



\maketitle

\section{Introduction}
\label{sec:introduction}

This paper is a foundational study of two styles of programming and their relationship:
\begin{enumerate}
  \item fresh name generation (gensym) via random draws;
  \item statistical probabilistic programming with higher-order functions.
\end{enumerate}
We use a recent model of probabilistic programming, quasi-Borel spaces (QBSs,~\cite{heunen:qbs}), to give a first random model of the $\nu$-calculus~\cite{stark:whatsnew}, which is a $\lambda$-calculus with fresh name generation.
By further developing the theory of QBSs, we are able to arrive at a new theorem for name generation:
\begin{theorem*}[\ref{thm:fullabstraction}]
  The random model of the $\nu$-calculus is fully abstract at first order. That is, two first order programs are observationally equivalent if and only if their interpretation in QBSs is the same.
\end{theorem*}
This is surprising because the simple \emph{non}-random models of the $\nu$-calculus, based on nominal sets~\cite[Ch.~9.6]{pitts:nominalsets} or functor categories~\cite[\S 5]{stark:cmln}, are \emph{not} fully abstract at first order~\cite[\S 5]{stark:cmln}.

\subsection{The \texorpdfstring{$\nu$}{Nu}-Calculus and its Observational Equivalence}
\label{sec:nucalc-intro}

The $\nu$-calculus (\S\ref{sec:names} and \cite{stark:whatsnew}) is a simply-typed $\lambda$-calculus with fresh name abstraction $\nu n.M$ in addition to $\lambda$-abstraction $\lambda x.M$. The idea is that $\nu n.M$ means ``generate a fresh name~$n$ and continue as $M$''. The $\nu$-calculus thus models name generation as used in various domains across computer science, including cryptography, distributed systems, and statistical modelling (see \S\ref{sec:relwork} for more background on name generation). Concretely, the $\nu$-calculus can also be viewed as a fragment of OCaml, where $\nu n.M$ abbreviates \lstinline|let n = ref() in M|, since a content-less reference is a pure name when there is no pointer arithmetic or comparison allowed.

The purpose of this paper is to give an interpretation of name generation in terms of randomness. The $\nu$-calculus already has a standard non-random operational semantics~\cite[\S2]{stark:whatsnew}, which induces a notion of observational equivalence $\approx$. For closed programs of ground type ($\tnamesintro,\tboolintro$), this is straightforward. For example, it includes the $\beta/\eta$ laws of the call-by-value $\lambda$-calculus, and also equations such as
\begin{equation}\label{eqn:namesdifferent}
  \nu m.\nu n.(m=n) \approx \false
\end{equation}
since any two separately generated names $m,n$ should be different.
Observational equivalence at first-order type ($\tnamesintro \to\tboolintro$, $\tboolintro\to\tboolintro\to\tnamesintro$, etc.), on the other hand,
is non-trivial in the $\nu$-calculus, because $\nu$'s and $\lambda$'s do not commute. For instance,
\begin{equation}\label{eq:lambdanu_commute}
  \nu n.\lambda x.n\not\approx \lambda x.\nu n.n\text.
\end{equation} So even at first order we can have complex nestings of $\nu$'s and $\lambda$'s.
In this paper we argue that a centerpiece of the first-order equational theory of the $\nu$-calculus is the following `privacy' equation~\cite[Ex.~4(2)]{stark:whatsnew}:
\begin{equation}
\label{eq:privnu}
  \nu n.\,\lambda x.(x=n)\ \ \approx\ \ \lambda x.\,\false \quad:\tnamesintro\to\tboolintro\text.
\end{equation}
On the left hand side, we generate a fresh name~$n$, and then return a function that takes an argument~$x$, and tests whether $x=n$. In this example, $n$ is chosen to be different from any name that the caller of the function knows, and the name is never revealed to the caller, and so, intuitively, it can never return $\true$. This is an example of an equation that is not validated by the standard nominal sets model, but it is validated by our QBS random model.

This aspect of name revelation is subtle, for instance, the program
\begin{equation}\label{eqn:call-twice}
  \nu m.\,\nu n.\,\lambda x.\,\ite{(x=m)} n m
\end{equation}
can reveal both $m$ and $n$, but it needs to be called twice to do this.
The random semantics takes care of this, as we explain.

\subsection{Probabilistic Programming and Name Generation as Randomness}
\label{sec:intro-probprog}

The idea of probabilistic programming (e.g.~\cite{mpyw-probprog-intro}) is to define complex probability distributions by writing programs. This is typically done by adding a \lstinline|sample| command to a $\lambda$-calculus, to allow primitive random draws. In the statistical setting, it is common to include continuous distributions over the real numbers, such as the normal distribution (Fig.~\ref{fig:normal-density}). For instance, the program
\begin{equation}\label{eqn:two-normal}
  \text{\lstinline|let x = sample(Normal(0,1)) in let y = sample(Normal(0,1)) in x+y|}
\end{equation}
is overall equivalent to sampling from a normal (Gaussian) distribution with mean~$0$ and variance~$2$.
The informal idea of this paper is to interpret $\nu n.M$ of the $\nu$-calculus as a probabilistic program:
\[\text{``}\nu n.\,M\ \ =\ \ \text{\lstinline|let n = sample(Normal(0,1)) in M|}\text{''}\]
so that freshly generated names are randomly sampled. A first observation is that any two draws from a normal distribution will almost surely be different, and so this interpretation validates~(\ref{eqn:namesdifferent}).

\begin{wrapfigure}[7]{r}{3.7cm}
  \pgfmathdeclarefunction{gauss}{2}{%
  \pgfmathparse{1/(#2*sqrt(2*pi))*exp(-((x-#1)^2)/(2*#2^2))}%
}\begin{tikzpicture}
\begin{axis}[
  domain=-3:3, samples=100,
  xlabel={}, ylabel={},
  axis lines=middle,
  ymax = 0.55,
  xmin = -4.2,
  xmax = 4.2,
  height=3cm, width=4.5cm,
  xtick={-4,-2,0,2,4}, ytick={0,0.5},
  ]
  \addplot [very thick,domain=-4.2:4.2,olive] {gauss(0,1)} ;
\end{axis}
\end{tikzpicture}
  \caption{Density of the normal distribution \lstinline|Normal(0,1)|.\label{fig:normal-density}}
\end{wrapfigure}
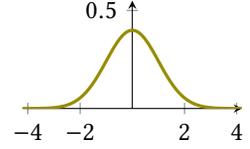
A probabilistic program involving sampling should be understood in terms of the histogram of results we see when we run the program a large number of times. To put it another way, the program~\eqref{eqn:two-normal} is a Monte Carlo description of the integral
$\iint k(x+y) \,\mathrm{d}y\,\mathrm{d}x$
where $\int$ denotes Lebesgue integration with respect to the normal probability measure and $k$ is some continuation function.
In this way, we may say, informally for now, that the random implementation of $\nu$-abstraction is also Lebesgue integration:
\[\textstyle\text{``}\nu n.\,M\ = \ \int M\,\mathrm{d}n\text{''}\]

As we will make precise in Sections~\ref{sec:intro-qbs} and~\ref{sec:probsemantics}, the measure-theoretic understanding of probability leads to full abstraction at first order. For a first glimpse, notice that in the $\nu$-calculus there is no definable function
\begin{equation}\label{eqn:exists}
  \exists :(\tnamesintro \to \tboolintro)\to \tboolintro
\end{equation}
such that $\exists (f)$ returns $\true$ if $f$ would ever return $\true$, as such a function would easily distinguish the programs in the privacy equation~\eqref{eq:privnu}. This function $\exists$ \emph{can} be defined in the nominal sets model (e.g.~\cite[\S2.5]{pitts:nominalsets}, \cite[eq.~2]{staton-local-state}), but is \emph{in}consistent with a measure-theoretic interpretation, as we now explain.
From this function $\exists$ we could easily define an expression
\[f:\tnamesintro \to \tnamesintro\to \tboolintro \ \ \vdash \ \ \lambda x.\,\exists(\lambda y.\,f\,x\,y)\ \  : \tnamesintro\to \tboolintro\]
which converts a subset of $(\tnamesintro \times \tnamesintro)$ to its existential projection as a subset of $(\tnamesintro)$.
In the setting of probability theory, we need to know that all definable expressions are measurable, so that integration can be used.
If we understand $(\tnamesintro)$ as the real numbers, and measurable subsets are Borel sets, as usual, then the projection of a Borel set is not necessarily Borel~\cite[14.2]{kechris}, and so the $\exists$ function~(\ref{eqn:exists}) cannot be in the model.
So our probabilistic interpretation of the $\nu$-calculus gives a new intuition for these privacy and definability issues.

\subsection{Quasi-Borel Spaces, Full Abstraction and Descriptive Set Theory}
\label{sec:intro-qbs}

A formalism that includes both measure theory and typed $\lambda$-calculus is quasi-Borel spaces (QBSs, \S\ref{sec:qbs} and \cite{heunen:qbs}). A QBS is a set $X$ together with a set of functions $M_X\subseteq[\mathbb{R}\to X]$ satisfying some conditions. The idea is to fix $\mathbb R$ as a source of randomness, and then $M_X$ describes the admissible random elements in $X$. For example,
for the QBS of booleans, we take $M_{\tboolintro}\subseteq [\mathbb R\to 2]$ to comprise the characteristic functions of Borel sets of $\mathbb R$, and for the QBS function space $[\trealintro\to \tboolintro]$ we take $M_{\trealintro\to \tboolintro}\subseteq [\mathbb R\to (\mathbb R\to 2)]$ to comprise the characteristic functions of Borel subsets of $\mathbb R^2$. In this way, we can interpret any $\nu$-calculus type as a QBS (\S\ref{sec:probsemantics}). Following the above discussion, we see that $\exists$ \eqref{eqn:exists} cannot be interpreted in QBSs.

We show our full abstraction theorem in this setting: two $\nu$-calculus programs of first-order type are observationally equivalent if and only if their interpretations in QBSs are equal (Thm.~\ref{thm:fullabstraction}). Our proof proceeds in three steps.
\begin{enumerate}
  \item We show that the privacy equation~(\ref{eq:privnu}) holds in QBS (\S\ref{sec:fullabstraction}). We have already mentioned that the $\exists$ function (\ref{eqn:exists}) cannot be defined in QBSs. The next step is to fully characterize the QBS space corresponding to $((\tnamesintro \to \tboolintro)\to \tboolintro)$. This turns out to correspond directly with the concept of `Borel-on-Borel' in descriptive set theory~\cite[\S18.B]{kechris}, and we use a pair of Borel inseparable sets to generalize the non-definability of $\exists$ and prove the privacy equation (Thm.~\ref{thm:privacy}).
  \item On the syntactic side, we give a normalization algorithm for observational equivalence at first order (\S\ref{sec:nf}, Thm.~\ref{thm:observational-equivalence-normal-form}). Our algorithm, which appears to be novel, refines a logical relations argument by Pitts and Stark~\cite{stark:whatsnew}, by identifying and eliminating all private names. This is non-trivial as, for instance,~(\ref{eqn:call-twice}) is already in normal form, but the similar program
  \[\nu m.\,\nu n.\,\lambda x.\,\ite{(x=m)} m n \qquad \text{ normalizes to } \qquad \nu n.\,\lambda x.\,n\text.\]
  Our construction simplifies the analysis of observational equivalence at first order (Thm.~\ref{thm:observational-equivalence-normal-form}). This also provides a general strategy for proving full abstraction (Thm.~\ref{thm:abstract-iff-normal-form-equal}).
  \item Returning to the semantic side, we show that the normalization steps are validated in the QBS model (\S\ref{sec:proof_fullabstraction}). The key idea here is that atomless measures such as the normal and uniform distributions are invariant under certain translations. We use this translation invariance to reduce our problem to the privacy equation~\eqref{eq:privnu}, and use this to prove full abstraction at first order (Thm.~\ref{thm:fullabstraction}). Our use of an invariant action on the space of names is similar to but distinct from nominal techniques~\cite[\S1.9]{pitts:nominalsets}; our action is internal to the model, and does not feature in its construction.
\end{enumerate}

In addition to proving full abstraction of the QBS semantics of the $\nu$-calculus at first order, we provide the first detailed investigation of the higher-typed function spaces in Borel-based probability theory (\S\ref{sec:privacy}, \S\ref{sec:structural}). The application of higher-order probabilistic methods is increasingly widespread in programming research (\S\ref{sec:related-hoprob} and \cite{sato:formal-verification,lew:tracetypes,plonk,ehrhard:cones,scibior}). We show that our programming-based development can alternatively be viewed in terms of recent categorical formulations of probability theory (\S\ref{sec:structural}). From this perspective, Bayesian inference (conditioning) is subtle in the higher-typed situation~(Prop.~\ref{prop:noconditionals}).
Intuitively, arbitrary conditioning
would mean that one could infer, from data as a function $(\tnamesintro\to \tboolintro)$, a posterior distribution on the names that the function privately uses, in violation of the privacy equation~(\ref{eq:privnu}).

In summary, through our full abstraction result (Thm.~\ref{thm:fullabstraction}), we formalize the relationship between random sampling and  fresh name generation, giving new perspectives on higher-order probability.

\section{Preliminaries on Name Generation and the \texorpdfstring{$\nu$}{Nu}-Calculus}
\label{sec:names}

In this section we recall the $\nu$-calculus~\cite{stark:whatsnew,stark:thesis}, which is a simple $\lambda$-calculus for name generation. We recall the syntax, the observational equivalence (\S\ref{sec:obseq}) and the denotational semantics (\S\ref{sec:names_catsemantics}).
Further discussion about name generation is in Section~\ref{sec:relwork}.

\begin{figure}
\small\centering
\begin{align*}
  \sigma,\tau &::= \tbool \,|\, \tnames \,|\, \sigma \to \tau
  \qquad\text{($\tbool$ and $\tnames$ abbreviate $(\mathsf{bool})$ and $(\mathsf{name})$ from \S\ref{sec:introduction} respectively.)}\\
M,N &::= x
\,|\, \true
\,|\, \false
\,|\, M = M
\,|\, M M
\,|\, \lambda x.M
\,|\, \nu n.M
\,|\, \ite M M M
\end{align*}
\vspace{-2mm}
\[
\infer[((x : \tau) \in \Gamma)]{\Gamma \vdash x : \tau}{} \qquad
\infer[(b = \true, \false)]{\Gamma \vdash b : \tbool}{}
\]
\vspace{-1mm}
\[
\infer{\Gamma \vdash \ite M {N_1} {N_2} : \tau}{\Gamma \vdash M : \tbool \quad \Gamma \vdash N_1 : \tau \quad \Gamma \vdash N_2 : \tau} \qquad
\infer{\Gamma \vdash (M=N) : \tbool}{\Gamma \vdash M : \tnames \quad \Gamma \vdash N : \tnames}
\]
\vspace{-1mm}
\[
\infer{\Gamma \vdash \nu x.M : \tau}{\Gamma,x\colon \tnames \vdash M:\tau} \qquad
\infer{\Gamma \vdash \lambda x.M : \sigma \to \tau}{ \Gamma, x:\sigma \vdash M : \tau} \qquad
\infer{\Gamma \vdash  M \, N : \tau}{\Gamma \vdash M : \sigma \to \tau \quad \Gamma \vdash N: \sigma}
\]
\caption{Grammar and typing rules for the $\nu$-calculus \cite[Table~1]{stark:whatsnew}.}
\Description{Grammar and typing rules for the $\nu$-calculus}
\label{fig:nucalculus}
\end{figure}

The types, syntax and typing judgements of the $\nu$-calculus are recalled in \cref{fig:nucalculus} \cite{stark:whatsnew}. The typing judgements are of the form $ \Gamma \vdash M : \tau$, where $\Gamma$ is a set of typed variables.

The types $\tbool,\tnames$ are called \emph{ground types}. Among higher types, we will pay special attention to \emph{first-order} function types, which are non-nested function types of the form  $\tau_1 \to \cdots \to \tau_n$ with each $\tau_i$ a ground type. As the $\nu$-calculus is call-by-value, first-order function types cannot be simplified by uncurrying and already contain considerable complexity. We elaborate this at the end of (\S\ref{sec:names_catsemantics}).



\subsection{Operational Semantics and Observational Equivalence}
\label{sec:obseq}

\begin{figure}
\small\centering
\[
  \infer{s \vdash V\Downarrow_\tau ()V }-
  \qquad
  \infer[m\neq n]{s\vdash (M=N)\Downarrow_\tbool (s_1\uplus s_2)\false}
  {s\vdash M\Downarrow_\tnames (s_1)m
    \quad
    s\vdash N\Downarrow_\tnames (s_2)n
  }
\qquad
  \infer{s\vdash (M=N)\Downarrow_\tbool (s_1\uplus s_2)\true}
  {s\vdash M\Downarrow_\tnames (s_1)m
    \quad
    s\vdash N\Downarrow_\tnames (s_2)m
  }
\]
\vspace{-1mm}
\[
\infer{s\vdash \ite M {N_\true}{N_\false} \Downarrow_\tau (s_1\uplus s_2)V'}
{s\vdash M\Downarrow_\tbool (s_1)V
  \quad
    (s\uplus s_1)\vdash N_V \Downarrow_\tau (s_2)V'}
  \qquad
  \infer[n\not\in s]{s\vdash \nu n.\,M\Downarrow_\tau (\{n\}\uplus s')V}
  {s\uplus\{n\}\vdash M\Downarrow_\tau (s')V}
\]
\vspace{-1mm}
\[
\infer{s\vdash M\, N\Downarrow_\tau (s_1\uplus s_2\uplus s_3)V'}
{s\vdash M\Downarrow_{\sigma \to \tau} (s_1)\lambda x.M'
  \quad
  (s\uplus s_1)\vdash N\Downarrow_{\sigma } (s_2)V
  \quad
 ( s\uplus s_1\uplus s_2)\vdash  M'[V/x]\Downarrow_\tau (s_3)V'
}\]
\caption{Evaluation relation for the $\nu$-calculus \cite[Table~2]{stark:whatsnew}.
  \label{fig:opsem}}
\Description{Evaluation relation for the $\nu$-calculus}
\end{figure}

The evaluation relation of the $\nu$-calculus is defined for terms with free variables of type $\tnames$, and no other free variables.
In this operational semantics, these variables are understood to be names that are generated in the course of running a program, and so they are assumed to be distinct, and we tend to use $m$ or $n$ for them.
If $s=\{n_1,\dots, n_k\}$ is a set of names and $\tau$ is a type, we define a set
\[\Exp_\tau(s)\ \defeq \ \Big\{M~|~n_1\colon \tnames,\dots n_k\colon \tnames \vdash M : \tau\Big\}\]
of expressions of type $\tau$ only involving the names $s$,
and we define the set $\Val_\tau(s)\subseteq \Exp_\tau(s)$ of values:
$\Val_\tau(s)=\{V\in\Exp_\tau(s)~|~V=\lambda x.M,\,V=\true,\,V=\false,\,V=n\}$.

If $s, t$ are sets of names, we write $s \oplus t$ to denote the \emph{disjoint union} of these names, which we can always form by renaming free names if necessary.

The big-step evaluation relation $s \vdash M \Downarrow_\tau (s')V$ is given in Figure~\ref{fig:opsem}, where $M \in \Exp_\tau(s)$ and $V \in \Val_\tau(s \oplus s')$, meaning $M$ evaluates to $V$ generating fresh names $s'$. Evaluation is terminating and deterministic up to choice of free names.
(We will not need to work directly with this evaluation relation very much in this paper, because we will build on existing methods for observational equivalence~\cite{stark:whatsnew,stark:cmln}, but we include it for completeness.)

Observational equivalence is defined in a standard way.
A \emph{boolean context $\mathcal C[\cdot]$ for type $\tau$}
is an expression $\mathcal C$ where some subexpressions are replaced by a placeholder,
such that if $M\in\Exp_\tau(s)$ then $\mathcal C[M]\in\Exp_\tbool(s)$.
Two terms $M_1,M_2\in \Exp_\tau(s)$ are \emph{observationally equivalent}, written $M_1 \approx_\tau M_2$, if for every boolean context $\mathcal{C}[\cdot]$ we have $\exists s'(s \vdash \mathcal{C}[M_1] \Downarrow_\tbool (s')\true) $ if and only if $ \exists s'(s \vdash \mathcal C[M_2] \Downarrow_\tbool (s')\true)$.

We have already given some examples of observational equivalences and inequivalences in Section~\ref{sec:nucalc-intro}. We illustrate the method a little more. To see that $\nu n.\lambda x.n \not \approx_{\tbool \to \tnames} \lambda x. \nu n. n$ (\ref{eq:lambdanu_commute}), consider the context $\mathcal C[-]=(\lambda f. (f\,\true) = (f\,\true))\,(-)$, which produces $\true$ for the first example and $\false$ for the right hand side.
On the other hand, an observational equivalence such as ${\nu n.\lambda x.(x=n)}\approx_{\tnames\to \tbool}\lambda x.\false$ (\ref{eq:privnu})
is a statement that quantifies over all contexts, and so requires a more elaborate method such as logical relations \cite[Example 5]{stark:whatsnew} or our random model (\S\ref{sec:privacy}).

We remark that the call-by-value semantics of the $\nu$-calculus form a central aspect of the intricacies of observational equivalence at first-order types. The $\lambda\nu$-calculus is a call-by-\emph{name} variation of the $\nu$-calculus \cite[\S9.4]{odersky:localnames,pitts:nominalsets}, and in that calculus, $\lambda$'s and $\nu$'s do commute \cite[Fig.~2]{odersky:localnames}, and then we can easily derive
\begin{equation}\label{eg:nu-lambda-commute}
  \nu n.\lambda x.(x=n) \approx_{\tnames \to \tbool} \lambda x.\nu n.(x=n) \approx_{\tnames \to \tbool} \lambda x.\false.
\end{equation}

\subsection{Categorical Semantics}
\label{sec:names_catsemantics}

The central definition of this paper is the random semantics of the $\nu$-calculus in Section~\ref{sec:probsemantics}.
Although this is a new semantics for the $\nu$-calculus, it is an instance of the very general categorical framework for $\nu$-calculus semantics given by Stark~\cite{stark:cmln}.
The rough idea is that one can interpret the $\nu$-calculus in any category with enough structure, by interpreting types as objects of the category and expressions as morphisms.

\paragraph{Metalanguages.}
This interpretation is clarified by using a metalanguage (aka internal language) to describe the morphisms of the category, and the way that they compose, instead of the traditional categorical composition notation (e.g.~\cite[\S I.10]{lambek-scott}). The metalanguage of cartesian closed categories allows us
to notate a morphism $A_1\times \dots \times A_n\to B$ as an expression
$x_1\colon A_1\dots x_n\colon A_n\vdash e :B$, and to use $\lambda$-notation and pairing to
manipulate the function spaces and products in the category.
Where the category also has a coproduct $1+1$, we can write the injections as $\vdash \true:1+1$ and $\vdash\false:1+1$, and the universal property of coproducts can be expressed in terms of an $\ite{/}{/}{}$ construction.
The interpretation of the $\nu$-calculus in a categorical model can be given by a translation from the $\nu$-calculus to this metalanguage.

\paragraph{Commutative Affine Monads.}
A strong monad $(T,[-],(-)^*)$ on a cartesian closed category $\mathbb C$
comprises an assignment of an object $T(A)$ for every object $A$ in $\mathbb C$,
a family of `return' morphisms $[-]:A\to T(A)$,
and a family of `bind' operations $(-)^*:T(B)^A\to T(B)^{T(A)}$, satisfying associativity and identity laws~\cite{moggi:computation_and_monads}.
In terms of the metalanguage, for any morphisms
described by expressions $\Gamma\vdash e : T(A)$
and $\Gamma,x\colon A \vdash e':T(B)$, we have a morphism
described by an expression
${\Gamma \vdash  \letin x e e'} : T(B)$ \cite{moggi:computation_and_monads}.
A strong monad is called
\emph{affine} and \emph{commutative} if the following \emph{discardability}~(\ref{eqn:affine}) and \emph{exchangeability}~(\ref{eqn:commutativity}) equations in the metalanguage are valid:
\begin{align}
  & \letin x e e'  \ = \ e' \hspace{4cm} \text{($x$ not free in $e'$)}\label{eqn:affine}\\
  & \begin{aligned}
      & \letin {x_1} {e_1}{\letin {x_2} {e_2} {e_3}} \ = \ \letin {x_2} {e_2}{\letin {x_1} {e_1} {e_3}}\\
      & \phantom{\letin x e e'  \ = \ e' \hspace{4cm} }\text{($x_1$ not free in $e_2$, $x_2$ not free in $e_1$).}
    \end{aligned}\label{eqn:commutativity}
\end{align}
Informally, affine means that we can discard any unused expressions, and is equivalent to $T(1)\cong 1$.
Commutativity means that we can exchange independent expressions~(e.g.~\cite{Kammar-Plotkin}).

\begin{definition}[\mbox{\cite[\S4.1]{stark:cmln}}]\label{def:catmodel}
  A \emph{categorical model of the $\nu$-calculus} comprises
  \begin{enumerate}
    \item a cartesian closed category $\mathbb C$ with finite limits;
    \item a strong monad $T$ on $\mathbb C$;
    \item a disjoint coproduct $B := 1+1$ of the terminal object with itself;
    \item a distinguished object of names $N$ with a decidable equality test $(=) : N \times N \to B$; and
    \item a distinguished morphism $\new : 1 \to T(N)$.
  \end{enumerate}
  We ask that this category satisfies the following additional axioms:
  \begin{enumerate}
    \item the monad $T$ is affine and commutative;
    \item the following equation holds in the metalanguage
      \begin{gather}
        \label{eq:freshness}
        m : N \vdash \letin n \new {[(n,m=n)]} \ = \ \letin n \new {[(n,\false)]} : T(N \times B). \tag{FRESH}
      \end{gather}
  \end{enumerate}
\end{definition}
The \eqref{eq:freshness} requirement allows us to reason within the metalanguage that any name generated with $(\new)$ is different from other names.
This definition references `disjoint coproducts' and `decidable equality', concepts from categorical logic, but we will not assume familiarity with these in the rest of the article except in the proof of Thm~\ref{thm:qbs-model}.

\paragraph{Denotational Semantics.}
In any categorical model we can interpret $\nu$-calculus types (Fig.~\ref{fig:nucalculus}) as objects, using the standard call-by-value translation into the monadic metalanguage: $\sem{\tbool} \defeq  B$, $\sem{\tnames} \defeq N$ and $\sem{\sigma \to \tau} \defeq \sem{\sigma} \to T\sem{\tau}$.
This is extended to contexts: $\sem\Gamma\defeq \prod_{(x\colon \tau)\in\Gamma}\sem \tau$.
A $\nu$-calculus expression $\Gamma\vdash M:\tau$ is routinely interpreted as a morphism
$\sem\Gamma \to T\sem \tau$ by induction on the structure of $M$ (Fig.~\ref{fig:densem}).

\begin{figure}\small
  \begin{align*}
    &\sem{\lambda x.M}\defeq [\lambda x.\sem M]
      \qquad
      \sem{x}\defeq [x]\qquad\sem{\true}\defeq [\true]\qquad\sem{\false}\defeq [\false]
    \\
    &
      \sem{M=N}\defeq \letin m {\sem M} {\letin  n {\sem N} {[m=n]}}
     \quad
    \sem{M\,N}\defeq \letin f {\sem M} {\letin  x {\sem N} {f(x)}}
     \\
    &
      \sem{\nu x.M}\defeq \letin x \new {\sem M}
    \quad\ 
    \sem{\ite M {N_1}{N_2}}
      \defeq \letin b {\sem M} \ite b {\sem {N_1}}{\sem{N_2}}
  \end{align*}
  \caption{Interpretation of $\nu$-calculus expressions in a categorical model, using its metalanguage~\cite[Fig.~5]{stark:cmln}. \label{fig:densem}}
  \Description{Interpretation of $\nu$-calculus expressions in a categorical model, using its metalanguage}
\end{figure}

Using the categorical limits and the equality test on $N$,
we can build a subobject $N^{\neq s} \rightarrowtail N^s$ for all finite sets $s$, modelling the assumption $(\neq s)$ of distinct names.
Formally, $N^{\neq s}$ is the equalizer of
$(n:N^s\vdash \bigvee_{i\neq j} (n_i= n_j):B)$ and $(n:N^s\vdash \false:B)$.
For expressions $M\in \Exp_\tau(s)$, we will typically use the restricted interpretation $\sem{M}_{\neq s} : N^{\neq s}\rightarrowtail N^s\xrightarrow {\sem M} T\sem{\tau}$.

We note that values $V\in \Val_\tau(s)$ factor through $[-]_{\sem{\tau}}:\sem{\tau}\to T\sem{\tau}$, i.e.~we can assume $\sem V:N^{\neq s}\to \sem\tau$. Intuitively, the values do not need a top-level monad because they do not generate fresh names.

Any categorical model according to Definition~\ref{def:catmodel} is sound and, under mild assumptions, adequate:
\begin{theorem}[{\cite[Prop.~1--4]{stark:cmln}}]\label{thm:sound-adequate}
  For any categorical model of the $\nu$-calculus:
  \begin{itemize}
    \item The big-step semantics is sound with respect to the denotational semantics: If $s \vdash M \Downarrow_\tau (s')V$ then $\sem{M}_{\neq s} = \sem{\nu s'.V}_{\neq s}$.
    \item If $1$ is not an initial object and $[-]_B:B\to T(B)$ is monic, then the denotational semantics is adequate for observational equivalence:
      If $\sem{M_1}_{\neq s} = \sem{M_2}_{\neq s}$ then $M_1 \approx_\tau M_2$, for all expressions $M_1, M_2\in\Exp_\tau(s)$.
  \end{itemize}
\end{theorem}
In Section~\ref{sec:nu} we survey the examples categorical models of the $\nu$-calculus from the literature. In Section~\ref{sec:probsemantics} we show that quasi-Borel spaces form a categorical model.

Categorical models need not identify observationally equivalent terms at higher types. The simplest example of such an equivalence is the privacy equation~\eqref{eq:privnu}, whose translation into the metalanguage is
\begin{align}
  \letin {a}{\new}{[\lambda x.[x = a]]} = [\lambda x.[\false]] : T\sem{\tnames \to \tbool} = T(N \to TB). \label{eq:rawpriv-c}
\end{align}

The metalanguage has extra types such as $(N\Rightarrow B)$ which are not
the interpretation of $\nu$-calculus types.
So in the metalanguage it is possible to consider the following simpler variation of \eqref{eq:rawpriv-c}:
\begin{align}
  \letin {a}{\new}{[\lambda x.(x = a)]} = [\lambda x.\false] : T(N \to B). \label{eq:priv-c}\tag{PRIV}
\end{align}
Note that \eqref{eq:priv-c} straightforwardly implies \eqref{eq:rawpriv-c} in the metalanguage.
So to prove the privacy observational equivalence \eqref{eq:privnu}, it is sufficient to find a categorical model that satisfies \eqref{eq:priv-c}.
In Section~\ref{sec:privacy} we show that quasi-Borel spaces satisfy \eqref{eq:priv-c}.
We discuss other models in Section~\ref{sec:nu}, in particular, neither~\eqref{eq:priv-c} nor~\eqref{eq:rawpriv-c} are satisfied in the nominal sets model (\ref{eqn:ctx-ex}).

We remark that because of the call-by-value semantics of $\nu$-calculus, first-order functions already exhibit an interesting degree of complexity that cannot be simplified by uncurrying. At type $\sem{\tnames \to (\tnames \to \tbool)} = T(\tnames \to T (\tnames \to T \tbool))$, name-generation effects may occur at three different stages, unlike in the uncurried version $T(\tnames \times \tnames \to T \tbool)$.

\section{Higher-Order Probability}
\label{sec:hoprob}

The central new definition of this paper is the random model of the
$\nu$-calculus based on quasi-Borel spaces.
We recall Borel spaces in Section~\ref{sec:measureable-spaces}, quasi-Borel spaces in Section~\ref{sec:qbs}, and then explain the model in Section~\ref{sec:probsemantics}, in prepration for the full abstraction result in Section~\ref{sec:fullabstraction}.

\subsection{Rudiments of Measurable Spaces}
\label{sec:measureable-spaces}

Probability spaces are traditionally defined in terms of measurable spaces \cite{kallenberg,pollard}. A \emph{measurable space} is a set $X$ together with a $\sigma$-algebra $\Sigma_X$ on $X$. We call a set $U \subseteq X$ \emph{measurable} if $U \in \Sigma_X$. A function $f : X \to Y$ between measurable spaces is \emph{measurable} if for all measurable $A \subseteq Y$, the set $f^{-1}(A)$ is measurable in $X$.

The measurable spaces and measurable functions form the category $\meas$. This category has products given by equipping $X \times Y$ with the product $\sigma$-algebra $\Sigma_X \otimes \Sigma_Y$.

The \emph{Borel $\sigma$-algebra $\Sigma_\R$} is the $\sigma$-algebra on $\R$ generated by the open intervals. We will always consider $\mathbb R$ as a measurable space with the Borel $\sigma$-algebra. We say a measurable space $X$ is \emph{discrete} if $\Sigma_X = \mathcal P(X)$, where $\mathcal P(X)$ denotes the power set of $X$.

A \emph{measure} on a measurable space $X$ is a $\sigma$-additive map $\mu : \Sigma_X \to [0, \infty]$ with $\mu(\emptyset) = 0$. It is \emph{finite} if $\mu(X) < \infty$, \emph{$s$-finite} if it is the countable sum of finite measures, and a \emph{probability measure} if $\mu(X) = 1$. A \emph{probability space} $(X, \mu)$ is a measurable space $X$ and a fixed probability measure $\mu$ on $X$.
If $\mu$ is a probability measure on $X$ and $f : X \to Y$ is measurable, then the \emph{pushforward measure} $f_*\mu$ on $Y$ is defined by $f_*\mu(U) = \mu(f^{-1}(U))$ for $U \in \Sigma_Y$. If $f\colon X\to \R$, then we can find the Lebesgue integral $\int_Xf(x)\,\dd \mu x\in \R$.

There is a monad $\mathcal G : \meas \to \meas$ due to \cite{giry} that assigns to $X$ the space of probability measures $\mathcal{G}X$ on $X$, with the $\sigma$-algebra generated by the maps $\mu \mapsto \mu(U)$ for all $U \in \Sigma_X$. The unit of this monad is the Dirac distribution $X \to \mathcal GX, x \mapsto \delta_x$. The bind of this monad consists of the averaging of measures, so that if $f: X \to \mathcal{G}Y$, we get the map $f^*: \mathcal{G}X \to \mathcal{G}Y$ taking $\mu \in \mathcal{G}X$ to the measure $f^*(\mu)(U) = \int_X f(x)(U)\,\dd\mu x$ on $Y$.
In the metalanguage, we can regard $\letin x \mu {f(x)}$ ($=f^*(\mu)$) as a generalized integral
$\int f(x)\,\dd \mu x$.
This monad is strong and commutative~(\ref{eqn:commutativity}), which is a categorical way to state Fubini's theorem \cite[1.27]{kallenberg}.
The monad is moreover affine~(\ref{eqn:affine}), since in general $g(y)=\int g(y)\,\dd\mu x$ for a probability measure~$\mu$.

When a probability space $(\Omega, \mu)$ is fixed, we say a \emph{random variable} $A$ with values in $X$ is a measurable map $A : \Omega \to X$. Two random variables $A,B$ are said to be \emph{equal in distribution}, written $A \overset{d}{=} B$, if they have the same law, i.e.~$A_*\mu=B_*\mu$ on $X$.

The spaces $\R$ and $[0, 1]$ are part of an important class of well-behaved measurable spaces called the standard Borel spaces. A \emph{standard Borel space} is a measurable space that is either countable and discrete or measurably isomorphic to $\R$ with the Borel $\sigma$-algebra. Note that this is not the usual definition of standard Borel spaces, which can be found in \cite[\S12.B]{kechris} and is equivalent to the one above. In particular, the definition of a standard Borel space ignores any underlying topology.

We refer to measurable subsets of standard Borel spaces as \emph{Borel sets}, measurable maps between standard Borel spaces as \emph{Borel measurable} and denote the full subcategory of standard Borel spaces by $\sbs$.

The standard Borel spaces form a well behaved full subcategory of $\meas$ closed under taking countable products and coproducts and the Giry monad. Additionally, Borel subsets of standard Borel spaces are standard Borel \cite[\S12.B, 13.4, 17.23]{kechris}.

Given a standard Borel space $X$, we call a probability measure $\mu$ on $X$ \emph{atomless} if $\mu(\{x\}) = 0$ for all $x \in X$. We have the following isomorphism theorem for standard Borel spaces with atomless probability measures:

\begin{theorem}[{\cite[17.41]{kechris}}]\label{thm:iso-meas}
  Let $\rho$ be the uniform measure on $[0, 1]$. If $X$ is a standard Borel space and $\mu$ an atomless measure on $X$, then there is a Borel measurable isomorphism $f: [0, 1] \to X$ such that $f_* \rho = \mu$.
\end{theorem}

\begin{example}\label{eg:atomless-sbs}
  The following are examples of familiar standard Borel spaces with atomless probability measures:
  \begin{enumerate}
    \item The space $\R$ of real numbers with the Gaussian distribution.
    \item The Cantor space $2^\N$, which can be viewed as the space of infinite sequences of coin flips, with the measure generated uniformly on the basic open sets: $\mu(\{s \in 2^\N : a \subseteq s\}) = 2^{-|a|}$, where $a$ is a finite sequence of flips.
    \item The circle $\T = [0, 1)$ (one-dimensional torus) with the uniform measure.
  \end{enumerate}
  By \cref{thm:iso-meas}, these are all isomorphic as probability spaces.
\end{example}

We note that a standard Borel space admitting an atomless probability measure is necessarily uncountable and in bijection with $\R$.

Measurable spaces are satisfactory for first-order probabilistic programming \cite{kozen-probprog,staton:sfinite}, but a result of Aumann shows that they fail to accommodate higher-order functions.

\begin{theorem}[Aumann {\cite{aumann}}]\label{thm:aumann}
  There is no $\sigma$-algebra on the space $2^\R$ of measurable functions $\R \to 2$ such that the evaluation map $2^\R \times \R \to 2$ is measurable.
\end{theorem}

We note that $2^\R$ can be identified with the set $\Sigma_\R$ of Borel sets in $\R$, and in this case the evaluation map $2^\R \times \R \to 2$ is simply the inclusion check $(B, x) \mapsto B \ni x$.

\subsection{Preliminaries on Quasi-Borel Spaces}
\label{sec:qbs}

Quasi-Borel spaces~\cite{heunen:qbs} are a convenient setting including both measure theory and higher-typed function spaces that are increasingly widely used (e.g.~\cite{scibior,lew:tracetypes,plonk,sato:formal-verification}).
They work by first restricting probability theory to the well-behaved domain of standard Borel spaces (\S\ref{sec:measureable-spaces}). They then provide a conservative extension to function spaces, achieving cartesian closure.
(We survey other models of higher-order probability in Section~\ref{sec:related-hoprob}.)

\begin{definition}[\cite{heunen:qbs}]
  A \emph{quasi-Borel space} is a set $X$ together with a collection $M_X$ of distinguished functions $\alpha : \R \to X$ called \emph{random elements}. The collection $M_X$ must satisfy
  \begin{enumerate}
    \item for every $x \in X$, the constant map $\lambda r.x$ lies in $M_X$;
    \item if $\alpha \in M_X$ and $\varphi : \R\to \R$ is Borel measurable, then $\alpha \circ \varphi \in M_X$; and
    \item if $\{A_i\}_{i=1}^\infty$ is a countable Borel partition of $\R$ and $\alpha_i \in M_X$ are given, then the case-split $\alpha(r) = \alpha_i(r)$ for $r \in A_i$ lies in $M_X$.
  \end{enumerate}
  A map $f : X \to Y$ between quasi-Borel spaces is a morphism if for all $\alpha \in M_X$ we have $f \circ \alpha \in M_Y$. This defines a category $\qbs$.
\end{definition}

We consider the reals with a canonical quasi-Borel structure $M_\R = \meas(\R,\R)$. Under that definition, any other quasi-Borel space $X$ satisfies $M_X = \qbs(\R,X)$. Similarly, we obtain a quasi-Borel structure on the space of booleans by taking $M_2 = \meas(\R,2)$ where $2$ is the two-point standard Borel space. This has the structure of a coproduct $2 \cong 1 + 1$.

The category $\qbs$ is cartesian closed, and we have $Y^X = \qbs(X,Y)$. By cartesian closure, a map $\R \to Y^X$ is a random element iff its uncurrying $\R \times X \to Y$ is a morphism.
For example, $2^\R$ comprises the characteristic functions of Borel subsets of $\R$, and the random elements $\R\to 2^\R$ are the curried characteristic functions of Borel subsets of $\R^2$.

Any quasi-Borel space $(X,M_X)$ can be equipped with a $\sigma$-algebra
$\Sigma_{M_X}=\qbs(X,2)$, where we identify subsets with their characteristic functions;
equivalently, $\Sigma_{M_X}$ is the greatest $\sigma$-algebra making the random elements measurable.

We now define probability theory in this new setting. Given a probability measure $\mu \in \mathcal G(\R)$ and $\alpha \in M_X$, we can push forward the randomness from $\R$ onto $X$, obtaining a distribution on $X$. The definition of the induced $\sigma$-algebra $\Sigma_{M_X}$ makes sure this pushforward is well-defined.

\begin{definition}[\cite{heunen:qbs}]
  A \emph{probability distribution} on a quasi-Borel space $X$ is an equivalence class $[\alpha,\mu]_\sim$, where $\alpha \in M_X, \mu \in \mathcal G(\R)$ and $(\alpha,\mu) \sim (\alpha',\mu')$ if $\alpha_*\mu = \alpha'_*\mu' \in \mathcal G(X,\Sigma_{M_X})$.
\end{definition}

We note that the significance of the induced $\sigma$-algebra on a quasi-Borel space $X$ is to give a notion of \emph{equality of distributions} on $X$, which is simply extensional equality of the pushforward measures.

There is a Giry-like strong monad $P$ on $\qbs$ which sends $X$ to the space $P(X)$ of probability distributions on $X$, endowed with the quasi-Borel structure
\begin{equation*}
  M_{P(X)} = \{\beta: \R \to P(X) \mid \exists \alpha \in M_X, g: \R \to \mathcal{G}\R ~\text{measurable s.t.}~ \beta(r) = [\alpha, g(r)]_\sim\}.
\end{equation*}
For $x \in X$, one can form the Dirac distribution $\delta_x$ on $X$ by taking $\delta_x = [\lambda r.x, \mu]_\sim$ for any $\mu \in \mathcal{G}\R$. This forms the unit of the monad. On the other hand, given $f: X \to P(Y)$ and $[\alpha, \mu]_\sim \in P(X)$, we have $f \circ \alpha \in M_{P(Y)}$ so there is some $\beta \in M_Y$ and $g: \R \to \mathcal{G}\R$ such that $f \circ \alpha(r) = [\beta, g(r)]_\sim$. We define a measure on $Y$ by taking $f^*([\alpha, \mu]_\sim) = [\beta, g^*(\mu)]_\sim$. This forms the bind of the monad.

Finally, we note that all of probability theory over standard Borel spaces is the same whether done in $\meas$ or $\qbs$.
\begin{proposition}[Conservativity {\cite[Prop. 19, 22]{heunen:qbs}}]\label{prop:conservativity}
  Any measurable space $(X,\Sigma_X)$ can be regarded as a quasi-Borel space $(X,M_{\Sigma_X})$, where $M_{\Sigma_X}=\meas(\R,X)$.
  This restricts to a full and faithful embedding $\sbs \to \qbs$ of standard Borel spaces into quasi-Borel spaces that preserves countable products, coproducts and the probability monad.
\end{proposition}

Due to this we will identify the standard Borel spaces in both $\meas$ and $\qbs$ and write say $2$ or $\R$ for the quasi-Borel space and measurable space alike.

Probability theory in $\qbs$ departs from the traditional foundations only if we go beyond standard Borel spaces. To emphasise this, we briefly make a digression to recall the categorical relationship between quasi-Borel spaces and measurable spaces.
\begin{proposition}[{\cite[Prop.~15]{heunen:qbs}}]\label{prop:qbs-meas-adjunction}
  The maps $\Sigma : (X, M_X) \mapsto (X, \Sigma_{M_X})$ and $M : (X, \Sigma_X) \mapsto (X, M_{\Sigma_X})$ are functorial and form an adjunction
  \begin{equation*}
    \xymatrix{
      \qbs \ar@/^/[r]_\bot^\Sigma & \meas \ar@/^/[l]^M
    }
  \end{equation*}
\end{proposition}

Now consider the quasi-Borel space $2^\R$. Using this adjunction (Prop.~\ref{prop:qbs-meas-adjunction}), we obtain a $\sigma$-algebra $\Sigma_{2^\R}$ on the set $2^\R$ and a measurable evaluation map $\Sigma(2^\R \times \R) \to 2$. We note that this does not contradict \cref{thm:aumann} because $\Sigma$ does not preserve products, and indeed the $\sigma$-algebra $\Sigma_{2^\R \times \R}$ induced from the quasi-Borel space $2^\R \times \R$ is strictly larger than the product algebra $\Sigma_{2^\R} \otimes \Sigma_\R$ (cf.~\cref{thm:privacy,prop:nu-nonpositive,obs:product-of-marginals}).

\subsection{Probabilistic Semantics for the \texorpdfstring{$\nu$}{Nu}-Calculus}
\label{sec:probsemantics}

We can now give probabilistic semantics to the $\nu$-calculus (cf. Def.~\ref{def:catmodel}) by interpreting names as elements of a probability space and name generation as random sampling.

\begin{theorem}\label{thm:qbs-model}
  $\qbs$ is a categorical model of the $\nu$-calculus under the following assignment:
  \begin{enumerate}
    \item the object of names is $N = \R$, and the object of Booleans is $B=2$;
    \item the name-generation monad is $T=P$; and
    \item $\new$ is given by the Gaussian distribution $\nu \in P(\R)$.
  \end{enumerate}
  Moreover, it is adequate: If $\sem{M_1}_{\neq s} = \sem{M_2}_{\neq s}$ then $M_1 \approx_\tau M_2$, for all expressions $M_1, M_2\in\Exp_\tau(s)$.
\end{theorem}

\begin{proof}
Quasi-Borel spaces have the required categorical structure, and the equality test is a Borel measurable map \mbox{$(=) : \R^2 \to 2$}, hence a morphism. The probability monad is commutative~(\ref{eqn:commutativity}), i.e.~Fubini holds~\cite[Prop.~22]{heunen:qbs}, and affine because  $P(1)\cong 1$, i.e.~probability measures must have total mass~$1$. The freshness requirement is the following identity in the internal language of $\qbs$, which reduces by Conservativity (Prop.~\ref{prop:conservativity}) to a statement about ordinary measure theory:
\begin{equation*}
  x : \R \vdash \letin y \nu {[(y, y = x)]} = \letin y \nu {[(y, \false)]} : P(\R \times 2)
\end{equation*}
Because $\nu$ is atomless, both sides denote the same distribution $\nu \otimes [\false]$.

For adequacy, we verify the assumptions of Thm.~\ref{thm:sound-adequate}. It is clear that $0\not\cong 1$. To see that the unit $[-]_B:B\to P(B)$ at~$B$ is monic, notice that by conservativity (Prop~\ref{prop:conservativity}) it is equivalent to check that $2\to \mathcal G(2)$ is injective in ordinary measure theory, which is trivial.
\end{proof}

\begin{remark}\label{rmk:any-sbs}
Any choice of standard Borel space and atomless measure will provide us with a model of the $\nu$-calculus. For example, we could consider $2^\N$ or $\T = [0, 1)$ with the uniform measure (cf. \cref{eg:atomless-sbs}), or $\R$ with any other atomless distribution.

By \cref{thm:iso-meas}, all such choices give isomorphic models of the $\nu$-calculus. More specifically, as the choice of standard Borel space and atomless measure completely determine the semantics of the $\nu$-calculus in $\qbs$, we always obtain the same equational theory of the $\nu$-calculus.

We may therefore choose to use any such space and measure when reasoning about the $\nu$-calculus in $\qbs$. We will take advantage of this in \cref{sec:proof_fullabstraction}, where we will find it convenient to work with the circle $\T$.
\end{remark}

By the general properties of categorical models, $\qbs$ semantics are sound and adequate for the $\nu$-calculus. In Section~\ref{sec:fullabstraction} we turn to studying the probabilistic semantics at higher types.

\paragraph{Aside on the `MONO' Requirement.}
When working with a monadic metalanguage, several authors~\cite{stark:cmln,moggi:computation_and_monads} ask that a monad~$T$ satisfies
the requirement
\begin{gather}\label{eq:mono} 
  \text{$[-]_X : X \to TX$ is monic for all $X$.}\tag{MONO}
\end{gather}
As we now explain, by using `separated' quasi-Borel spaces we can support the full \eqref{eq:mono} requirement. We mention this for completeness with respect to the literature, and will not use this notion later in this paper. In Stark's adequacy result~(Thm.~\ref{thm:sound-adequate}(2)), he only requires that $[-]_B:B\to T(B)$ be monic (for $B=1+1$).
\begin{definition}
  A quasi-Borel space $(X,M_X)$ is \emph{separated} if the maps $X \to 2$ separate points, meaning that for all $x \neq x' \in X$ there is some morphism $f: X \to 2$ such that $f(x) \neq f(x')$.
\end{definition}

This is equivalent to saying that the induced $\sigma$-algebra $\Sigma_{M_X}$ on $X$ separates points.

\begin{proposition}
  A quasi-Borel space $X$ is separated if and only if it satisfies the (\ref{eq:mono}) rule: the unit $X \to P(X)$ of the probability monad is injective.

  Additionally, we have:
  standard Borel spaces are separated;
  if $X,Y$ are separated, so is $X \times Y$;
  if $Y$ is separated, so is $Y^X$; and
  for every $X$, $P(X)$ is separated.
\end{proposition}

\begin{proof}[Proof notes]
  The first part follows because for $f: X \to 2$ and $x \in X$, we have $\int_X f(y) \,\dd{\delta_x}y = f(x)$.
  The rest is routine calculation.
\end{proof}
Therefore we could model the full \eqref{eq:mono} requirement by restricting to \emph{separated} quasi-Borel spaces. Moreover, this would not change the semantic interpretation.

\section{Full Abstraction}
\label{sec:fullabstraction}

In this section, we will prove that $\qbs$ is a fully abstract model of the $\nu$-calculus at first-order types. This will proceed in three steps, as described in \S\ref{sec:intro-qbs}. We will first prove that privacy holds in $\qbs$ (\S\ref{sec:privacy}). We will then construct a normal form invariant observational equivalence at first-order types, eliminating the use of private names (\S\ref{sec:nf}). Finally, we will make use of a measure-invariant group structure on the set of names and the privacy equation established in \S\ref{sec:privacy} to prove that $\qbs$ validates our normalization and is therefore fully abstract at first-order types (\S\ref{sec:proof_fullabstraction}).

\subsection{The Privacy Equation}
\label{sec:privacy}

\begin{theorem}[Privacy for $\qbs$]\label{thm:privacy}
  $\qbs$ satisfies \eqref{eq:priv-c}. This means that the random singleton is indistinguishable from the empty set:
  \begin{equation*}
    \letin a \nu {[\{a\}]} = [\emptyset] : P(2^\R).
  \end{equation*}
  In particular, $\qbs$ validates the privacy equation~\eqref{eq:privnu}.
\end{theorem}

In statistical notation, we would consider a Borel set-valued random variable $\{X\}$ where $X \sim \nu$. Privacy states that $\{X\} \overset{d}{=} \emptyset$ in distribution. Before presenting the proof, let us consider some examples of measurable operations which we can apply to Borel sets and see why they fail to distinguish $\{X\}$ from $\emptyset$.

\begin{example}\label{eg:inclusion-rand}
  For any fixed number $x_0 \in \R$, the evaluation map $x_0 \in (-) : 2^\R \to 2$ is a morphism. However, testing membership of $x_0$ will almost surely not distinguish $\{X\}$ and $\emptyset$, as $X$ is sampled from an atomless distribution, so
  \begin{align*}
    \P(x_0 \in \{X\}) = \P(X = x_0) = 0 = \P(x_0 \in \emptyset).
  \end{align*}
  This is merely stating freshness: a freshly generated name is distinct from any fixed existing name. As discussed in \cref{eg:nu-lambda-commute}, this is a strictly weaker statement than privacy, because $\lambda$ and $\nu$ don't commute.
\end{example}

\begin{example}\label{eg:s-finite-rand}
  \Cref{eg:inclusion-rand} shows that Dirac distributions cannot distinguish the random singleton from the empty set. More generally, they cannot be distinguished by $s$-finite measures. Evaluating an $s$-finite measure $\mu$ is a morphism $2^\R \to [0, \infty]$ {\cite[\S 4.3]{scibior}}. However because the set of atoms of $\mu$ is countable, we have $\mu(\{X\}) = 0 = \mu(\emptyset)$ almost surely.
\end{example}

\begin{example}\label{ex:qbs-nonemptiness}
  In Section~\ref{sec:intro-probprog} we discussed the Boolean existence function~(\ref{eqn:exists}), recalling that if it was in a model then the privacy equation~(\ref{eq:privnu}) would not hold. As we suggested, this function is incompatible with Borel-based probability. We can now be precise: the nonemptiness check $\exists: 2^\R \to 2$ is not a quasi-Borel morphism.

  To see that this is the case, recall that there exists a Borel subset $B \subseteq \R^2$ of the plane whose projection $\pi(B)$ is not Borel \cite[14.2]{kechris}. The characteristic function $\chi_B : \R \times \R \to 2$ is a morphism, and so is its currying $\beta : \R \to 2^\R$. However, the characteristic function $\chi_{\pi(B)} = \exists \circ \beta$ is not a morphism because $\pi(B)$ is not measurable. Therefore $\exists : 2^\R \to 2$ cannot be a quasi-Borel map.

  This implies that the singleton $\{\emptyset\} \subseteq 2^\R$ is not measurable. Furthermore, the equality check between sets $2^\R \times 2^\R \to 2$ is not a morphism in $\qbs$.
\end{example}

As \cref{thm:privacy} is a statement about measures on $2^\R$, we must analyze the $\sigma$-algebra $\Sigma_{2^\R}$ on $2^\R$ induced by its quasi-Borel structure.

\begin{notation}
  Let $B \subseteq X \times Y$ and $x \in X$. We let $B_x = \{y \in Y \mid (x, y) \in B\}$ denote the vertical section of $B$ at $x$.
\end{notation}

Recall that we can identify the space $2^\R = \qbs(\R, 2)$ with the Borel subsets of $\R$. We can similarly identify the set $\qbs(\R \times \R, 2)$ with the Borel subsets of $\R \times \R$, and by currying this means that the maps in $\qbs(\R, 2^\R)$ are exactly the maps $\lambda r.B_r$ for Borel $B \subseteq \R \times \R$. If $B \subseteq \R \times \R$ and $\mathcal{U} \subseteq 2^\R$, we note that
\begin{equation*}
  (\lambda r.B_r)^{-1}(\mathcal{U}) = \{r \in \R \mid B_r \in \mathcal{U}\}.
\end{equation*}

\begin{definition}[{\cite{kechris}}]
  A collection $\mathcal{U} \subseteq 2^\R$ of Borel sets is \emph{Borel on Borel} if for all Borel $B \subseteq \R \times \R$, the set $\{r \in \R \mid B_r \in \mathcal{U}\}$ is Borel.
\end{definition}

It follows that the $\sigma$-algebra $\Sigma_{2^\R}$ on $2^\R$ induced by the quasi-Borel structure is exactly the collection of Borel on Borel sets. Examples of such families include the family of null sets with respect to a Borel probability measure (\cref{eg:s-finite-rand}) and the family of meager sets \cite[\S18.B]{kechris}.

\begin{definition}[{\cite{kechris}}]
  Let $X$ be a standard Borel space. Two disjoint sets $A, A' \subseteq X$ are said to be \emph{Borel separable} if there is a Borel set $B \subseteq X$ such that $A \subseteq B$ and $A' \cap B = \emptyset$. $A, A'$ are \emph{Borel inseparable} if no such set exists.
\end{definition}

\begin{theorem}[Becker {\cite[35.2]{kechris}}]\label{thm:becker}
  There exists a Borel set $B \subseteq \R \times \R$ such that the sets
  \begin{equation*}
    B^0 = \{x \in \R \mid B_x = \emptyset\} \quad\text{and}\quad B^1 = \{x \in \R \mid B_x ~\text{is a singleton}\}
  \end{equation*}
  are Borel inseparable.
\end{theorem}

Using this, we prove that quasi-Borel spaces validate privacy.

\begin{lemma}\label{lemma:bob-singleton}
  Let $\mathcal{U} \subseteq 2^\R$ be Borel on Borel. If $\emptyset \in \mathcal{U}$ then $\{r\} \in \mathcal{U}$ for all but countably many $r \in \R$.
\end{lemma}

\begin{proof}
Let $A = \{r \in \R \mid \{r\} \notin \mathcal{U}\}$. This is a Borel set because $\mathcal{U}$ is Borel on Borel. Borel subsets of standard Borel spaces are standard Borel, so $A$ is standard Borel.

Now suppose for the sake of contradiction that $A$ were uncountable. Because $A$ is standard Borel it is isomorphic to $\R$. Fixing such an isomorphism, we have by \cref{thm:becker} a Borel set $B \subseteq \R \times A$ such that $B^0, B^1$ are Borel inseparable.

However, if $r \in B^0$ then $B_r = \emptyset \in \mathcal{U}$. On the other hand, if $r \in R^1$ then $B_r = \{a\}$ for some $a \in A$, and so $B_r = \{a\} \notin \mathcal{U}$. It follows that $B^0 \subseteq \{r \in \R \mid B_r \in \mathcal{U}\}$ and $B^1 \subseteq \{r \in \R \mid B_r \notin \mathcal{U}\}$. As $\mathcal{U}$ is Borel on Borel, $\{r \in \R \mid B_r \in \mathcal{U}\}$ provides a Borel separation of $B^0, B^1$, a contradiction.
\end{proof}

\begin{proof}[Proof of \cref{thm:privacy}]
To show that these two quasi-Borel measures are equal, we must check that the pushforward measures agree on the measurable space $(2^\R, \Sigma_{2^\R})$, meaning that for $\mathcal{U} \in \Sigma_{2^\R}$,
\begin{equation*}
\emptyset \in \mathcal{U} \iff \nu \{r \in \R \mid \{r\} \in \mathcal{U}\} = 1.
\end{equation*}
Every such $\mathcal{U}$ is Borel on Borel, and by possibly taking complements we can assume that $\emptyset \in \mathcal{U}$. By \cref{lemma:bob-singleton} the set $\{r \in \R \mid \{r\} \in \mathcal{U}\}$ is co-countable, and because $\nu$ is atomless this must have $\nu$-measure $1$.
\end{proof}

We offer some comments about this proof: the strategy we employed generalizes beyond the category of quasi-Borel spaces. Take any model of higher-order probability which agrees with standard Borel spaces on ground types, that is all morphisms $\R \to 2$ are measurable and all measurable maps $\R^2 \to 2$ are morphisms. Then this Borel on Borel property is a necessary constraint on second-order functions $2^\R \to 2$, arising from cartesian closure alone. In this case, \cref{lemma:bob-singleton} applies and it is inconsistent for such morphisms to tell apart the empty set from a random singleton with positive probability.

It is now merely an extensionality aspect of $\qbs$ that these constraints are also sufficient, and that the inability to distinguish the empty set from singletons implies equality in distribution. The category of sheaves in \cite{staton:sheaves} features a more intensional probability monad, where the two sides of the privacy equation presumably cannot be identified.

\subsection{A Normal Form for Privacy}
\label{sec:nf}

The privacy equation is a crucial stepping stone to full abstraction at first-order types. In \cref{sec:proof_fullabstraction} we will show that all other first-order observational equivalences can be reduced to privacy. In order to do this, we will first define a syntactic procedure to eliminate private names. Intuitively, private names are names that are not leaked to the environment --- if they are not already known outside the program, then they have no observable effects. In this section, we will provide a concrete definition of private names in terms of a logical relation originally developed in \cite{stark:whatsnew}, and we will construct a normal form invariant under observational equivalence that eliminates the use of private names in first-order terms.\\

\begin{figure}
   \centering
    \[
    b_1 \,R^\val_\tbool\, b_2 \Leftrightarrow b_1 = b_2 \qquad n_1 \,R^\val_\tnames\, n_2 \Leftrightarrow n_1 \,R\, n_2
    \]
    %
    %
    \begin{align*}
        (\lambda x.M_1)\,R^\val_{\sigma \to \tau}\,(\lambda x.M_2) \Leftrightarrow
        \forall R' &\colon s'_1 \leftrightharpoons s'_2, V_1 \in \Val_\sigma(s_1\oplus s'_1), V_2 \in \Val_\sigma(s_2 \oplus s'_2), \\
        &V_1 \,(R \oplus R')^\val_\sigma\, V_2 \Rightarrow M_1[V_1/x] \,(R \oplus R')^\exp_\tau\, M_2[V_2/x]
    \end{align*}
    %
    %
    \begin{align*}
        M_1 \,R^\exp_\tau\, M_2 \Leftrightarrow \exists R' &\colon s'_1 \leftrightharpoons s'_2, V_1 \in \Val_\sigma(s_1 \oplus s'_1), V_2 \in \Val_\sigma(s_2 \oplus s'_2), \\
        & s_1 \vdash M_1 \Downarrow_\sigma (s'_1)V_1 \,\&\, s_2 \vdash M_2 \Downarrow_\sigma (s'_2)V_2 \,\&\, V_1 \,(R \oplus R')^\val_\sigma\, V_2
    \end{align*}
    \caption{Stark's logical relation}
    \label{fig:logicalrelation}
    \Description{Stark's logical relation}
\end{figure}

Let $s_1,s_2$ be sets of free names; we write $R \colon s_1 \leftrightharpoons s_2$ for a partial bijection or \emph{span} between $s_1$ and $s_2$. We write $R \oplus R'$ for the disjoint union of spans between disjoint sets of names, and we write $\id_s\colon s \oplus t_1 \leftrightharpoons s \oplus t_2$ to denote the partial bijection defined that is the identity on $s$ and undefined on $t_1, t_2$. Stark \cite{stark:whatsnew} defines two families of relations $R^\val_\tau \subseteq \Val_\tau(s_1) \times \Val_\tau(s_2)$ and $R^\exp_\tau \subseteq \Exp_\tau(s_1) \times \Exp_\tau(s_2)$ by mutual induction, given in \cref{fig:logicalrelation}.

We note that $R^\val_\tau$ and $R^\exp_\tau$ coincide at values, so we will simply write the relations as $R_\tau$. Additionally, by renaming related names we can without loss of generality reduce any span $R$ to a subdiagonal, writing $s_i = s \oplus t_i$ and $R=\id_s$.

The logical relation agrees with observational equivalence $(\approx)$ at first-order types:

\begin{theorem}[{\cite[Theorem 22]{stark:whatsnew}}]\label{thm:logical-relation-observational-equivalence}
  Let $\tau$ be a first-order type. Then for $M_1, M_2 \in \Exp_\tau(s)$ we have
  \begin{equation*}
    M_1 \approx_\tau M_2 \Leftrightarrow M_1 \,(\id_s)_\tau\,M_2
  \end{equation*}
\end{theorem}

It is important to note that the logical relation is defined at all types $\tau$, but the relation at first-order types need only quantify over smaller first-order or ground types, making it possible to reason about observational equivalence of such terms inductively. In this paper we will primarily focus on the logical relation at first order types. In this setting we can tighten up Theorem 4.10 further, as we will explain: for any $s'\subseteq s$, $(\id_{s'})_\tau$ is a partial equivalence relation whose domain comprises those expressions that don’t leak any names when $s'$ is public, and $(\id_{s'})_\tau$ relates expressions whose behaviours are equivalent when $s'$ is public.

\begin{example}\label{eg:rel-priv}
  The privacy equation for the $\nu$-calculus~(\ref{eq:privnu}) can be established by means of this logical relation. Because $\{a, x\} \vdash (x = a) \Downarrow_\tbool \false$ whenever $a, x$ are distinct names, the logical relation implies that
  \begin{equation*}
    \lambda x. (x=a) \,(\id_\emptyset)_{\tnames \to \tbool}\, \lambda x.\false,
  \end{equation*}
  so that intuitively $a$ is private in $\lambda x. (x=a)$. This in turn implies that
  \begin{equation*}
    \nu a . \lambda x. (x=a) \,(\id_\emptyset)_{\tnames \to \tbool}\, \lambda x.\false,
  \end{equation*}
  which by \cref{thm:logical-relation-observational-equivalence} establishes the privacy equation of the $\nu$-calculus.
\end{example}

\begin{example}\label{eg:rel-trans}
  For names $a, b$, let $\lambda x.(a\,b)x: \tnames \to \tnames$ denote the term
  \begin{equation*}
    \lambda x.\ite {(x = a)} b \ite {(x = b)} a x.
  \end{equation*}
  This is the transposition of $a, b$, swapping $a$ and $b$ and otherwise behaving as the identity. It is clear that $\lambda x.(a\,b)x \,(\id_{\{a, b\}})_{\tnames \to \tnames}\, \lambda x.(a\,b)x$. One can also verify that $\lambda x.(a\,b)x \,(\id_\emptyset)_{\tnames \to \tnames}\, \lambda x.(a\,b)x$ as well. Here we no longer allow relations to be made with the names $a, b$, which we think of as private. Similarly, one can check that $\lambda x.(a\,b)x \,(\id_\emptyset)_{\tnames \to \tnames}\, \lambda x.x$, so that
  \begin{equation*}
    \nu a.\nu b.\lambda x.(a\,b)x \,(\id_\emptyset)_{\tnames \to \tnames}\, \lambda x.x
  \end{equation*}
  and by \cref{thm:logical-relation-observational-equivalence} $\nu a.\nu b.\lambda x.(a\,b)x$ is observationally equivalent to the identity.

  We note that it is not the case that $\lambda x.(a\,b)x \,(\id_{\{a\}})_{\tnames \to \tnames}\, \lambda x.(a\,b)x$, as this would require that $b \,(\id_{\{a\}})\, b$. The same holds if we swap $a$ for $b$. It is therefore apparent that the logical relations capture some of the connections between names; in this case, that if $a$ or $b$ are known, then by passing them as an argument to $\lambda x.(a\,b)x$ the other will be made public as well.
\end{example}

We notice that in these examples, private names are unmatched by spans. Intuitively, this is because the unmatched names do not affect the (observational) semantics of the terms; if we do not already know what they are, then they have no observable effects. In general, given $M \in \Exp_\tau(s \oplus t)$, we are interested in the names in $t$ with observable effects given that the names in $s$ are known to the environment. This motivates \cref{def:private-names} of private and leaked names.

\begin{notation}
  If $R\colon s_0 \leftrightharpoons s_1$ and $S\colon s_1 \leftrightharpoons s_2$ are spans, we let $R; S$ denote the \emph{composition of relations}, meaning that $m (R; S) n$ if there is some $z$ such that $m R z$ and $z S n$.
\end{notation}

\begin{lemma}\label{lem:transitivity}
  The logical relations are transitive at first-order types. This means that if $\sigma$ is a first-order type, $M_i \in \Exp_\sigma(s_i)$ for $i = 0, 1, 2$ and $R\colon s_0 \leftrightharpoons s_1, S\colon s_1 \leftrightharpoons s_2$ are spans such that $M_0 \,R_\sigma\, M_1$ and $M_1 \,S_\sigma\, M_2$, then $M_0 \,(R; S)_\sigma\, M_2$.
\end{lemma}

\begin{proof}
  This follows by induction on the type $\sigma$.
\end{proof}

\begin{proposition}\label{prop:minimal-relation}
  Let $\sigma$ be a first-order type and $M \in \Exp_\sigma(s \oplus t)$. There is a least $u \subseteq t$ such that $M \,(\id_{s \oplus u})_\sigma\, M$.
\end{proposition}

\begin{proof}
  If $u_0, u_1 \subseteq t$, $M \,(\id_{s \oplus u_0})_\sigma\, M$ and $M \,(\id_{s \oplus u_1})_\sigma\, M$, then $\id_{s \oplus u_0}; \id_{s \oplus u_1} = \id_{s \oplus (u_0 \cap u_1)}$ so by transitivity (\ref{lem:transitivity}) we have $M \,(\id_{s \oplus (u_0 \cap u_1)})_\sigma\, M$. We can therefore take $u$ to be the intersection of all such sets.
\end{proof}

\begin{proposition}\label{prop:logical-relations-minimality}
  Let $\sigma$ be a first-order type. Let $M_i \in \Exp_\sigma(s \oplus t_i)$ and suppose there is some $R\colon t_1 \leftrightharpoons t_2$ such that $M_1 \,(\id_s \oplus R)_\sigma\, M_2$. Let $u_i \subseteq t_i$ be the least set such that $M_i \,(\id_{s \oplus u_i})_\sigma\, M_i$. Then after possibly renaming names in $u_i$ we have $u_1 = u_2 = u$, $\id_u \subseteq R$ and $M_1 \,(\id_{s \oplus u})_\sigma\, M_2$.
\end{proposition}

\begin{proof}
We know that $R; R^{-1} = \id_{\dom(R)}$, so $M_1 \,(\id_{s \oplus dom(R)})_\sigma\, M_1$ by transitivity (\ref{lem:transitivity}). As $u_1$ is least with this property, $u_1 \subseteq dom(R)$.

Now consider the restriction $R\restriction_{u_1}$ of $R$ to $u_1$. Because $R\restriction_{u_1} = \id_{u_1}; R$, we have by transitivity that $M_1 \,(\id_s \oplus R\restriction_{u_1})_\sigma\, M_2$.

A symmetric argument shows that $R\restriction_{u_1}$ is a bijection of $u_1$ onto $u_2$. Therefore, after renaming names, we can assume that $u_1 = u_2 = u$ and $R\restriction_{u} = \id_u$.
\end{proof}

\begin{definition}[Private and Leaked Names]\label{def:private-names}
  Let $M \in \Exp_\tau(s \oplus t)$. We define \emph{the set of leaked names in $M$ relative to $s$}, denoted by $\Leak(M, s)$, to be the least $u \subseteq t$ such that $M \,(\id_{s \oplus u})_\tau\, M$. We call the names that are not leaked \emph{private relative to $s$}, denoted $\Priv(M, s) = t \setminus \Leak(M, s)$. Given a type $\tau$ and a set of names $s$, we let
   \[\Safe^s_\tau = \{M \in \Exp_\tau(s \oplus t) \mid \Leak(M, s) = \emptyset\} = \{M \in \Exp_\tau(s \oplus t) \mid M \,(\id_s)_\tau\, M\}\]
   be the set of terms that leak no names relative to $s$. If $s$ is empty, we write $\Priv(M)$, $\Leak(M)$ and $\Safe_\tau$.
\end{definition}

\begin{remark}
By \cref{lem:transitivity,prop:logical-relations-minimality}, the relation $(\id_s)_\tau$ induces an equivalence relation on $\Safe^s_\tau$. In fact, this corresponds to the usual notion of reducibility by logical relations. Concretely, one could equivalently define $\Safe^s_\tau$ directly as follows:
\begin{gather*}
\true, \false \in \Safe^s_\tbool \qquad n \in \Safe^s_\tnames \Leftrightarrow n \in s\\
\lambda x.M \in \Safe^s_{\sigma \to \tau} \Leftrightarrow \forall s', V \in \Safe^{s \oplus s'}_\sigma, M[V/x] \in \Safe^{s \oplus s'}_\tau\\
M \in \Safe^s_\tau \Leftrightarrow \exists s', V \in \Safe^{s \oplus s'}_\tau, M \Downarrow (s')V.
\end{gather*}
\end{remark}

\begin{example}\label{eg:private-names}
We have $\Priv(\lambda x.(x = a)) = \{a\}$ and $\lambda x.(x = a) \in \Safe_{\tnames \to \tbool}$ (cf. \cref{eg:rel-priv}). Similarly, $\Priv(\lambda x.(a\,b)x) = \{a, b\}$ and $\lambda x.(a\,b)x \in \Safe_{\tnames \to \tnames}$ (cf. \cref{eg:rel-trans}).
\end{example}

In \cref{eg:rel-priv,eg:rel-trans,eg:private-names}, we identified private names and found logically related terms that eliminate them. We will now show that this can be done for all terms of first-order type by constructing a normal form that recursively eliminates private names.

\begin{notation}
  If $s = \{n_1, \dots, n_k\}$ is a set of names, we write $\nu s.M$ as shorthand for $\nu n_1.\dots.\nu n_k.M$. We also write
  \begin{equation*}
    \ite {x = n \in s} {M_n} {M_0}
  \end{equation*}
  as shorthand for
  \begin{equation*}
    \ite {x = n_1} {M_{n_1}} {\mathtt{if}}~\cdots~{\mathtt{else}}~\ite {x = n_k} {M_{n_k}} {M_0}.
  \end{equation*}
\end{notation}

\begin{definition}[Normal form for privacy]
Let $\sigma$ be a first-order type and let $M \in \Safe^s_\sigma$ for $M \in \Exp_\tau(s \oplus t)$. We define the \emph{normal form} $\norm{M, s}$ of $M$ by induction on the type $\sigma$ as follows:

\textbf{Ground case:} If $\sigma$ is a ground type and $M$ is a value, then we let $\norm{M, s} = M$.

\textbf{Function case $\tbool \to \tau$:} Suppose $M$ is a value of type $\tbool \to \tau$ and that we have already constructed normal forms for expressions of type $\tau$. Expanding $M$ into its $\eta$-normal form, we have
\begin{equation*}
M = \lambda x.\ite {x = \true} {M_\true} {M_\false}
\end{equation*}
for some $M_\true, M_\false \in \Exp_\tau(s \oplus t)$. We know $M \,(\id_s)_\sigma\, M$, so we have that $M_\true \,(\id_s)_\tau\, M_\true$ and $M_\false \,(\id_s)_\tau\, M_\false$. We then define
\begin{equation*}
\norm{M, s} = \lambda x.\ite {x = \true} {\norm{M_\true, s}} {\norm{M_\false, s}}.
\end{equation*}

\textbf{Function case $\tnames  \to \tau$:} Suppose that $M$ is a value of type $\tnames \to \tau$ and that we have already constructed normal forms for expressions of type $\tau$. Expanding $M$ to its $\eta$-normal form, we have
\begin{equation*}
M = \lambda x.\ite {x = n \in s \oplus t} {M_n} {M_0}
\end{equation*}
for some $M_n \in \Exp_\tau(s \oplus t)$ and $M_0 \in \Exp_\tau(s \oplus t \oplus \{x\})$. In this case, $M \,(\id_s)_\sigma\, M$ implies that $M_0 \,(\id_{s \oplus \{x\}})_\tau\, M_0$ and $M_n \,(\id_s)_\tau\, M_n$ for all $n \in s$. We then define
\begin{equation*}
\norm{M, s} = \lambda x.\ite {x = n \in s} {\norm{M_n, s}} {\norm{M_0, s \oplus \{x\}}}.
\end{equation*}

\textbf{Expression case:} Suppose that we have constructed normal forms for values of type $\sigma$. Because $M \,(\id_s)_\sigma\, M$, there is some $V \in \Val_\sigma(s \oplus t \oplus w)$ such that $s \oplus t \vdash M \Downarrow_\sigma (w)V$ and $V \,(\id_{s \oplus w'})_\sigma\, V$ for some $w' \subseteq w$. Let $u = \Leak(V, s) \subseteq w'$. Then $V \,(\id_{s \oplus u})_\tau\, V$, so we can define
\begin{equation*}
\norm{M, s} = \nu u.\norm{V, s \oplus u}.
\end{equation*}
\end{definition}

If $s$ is empty, we omit it and write $\norm{M}$.

\begin{example}\label{eg:normal-forms}
  This normal form generalizes the observations in \cref{eg:rel-priv,eg:rel-trans}. Specifically, we have
  \begin{equation*}
    \norm{\nu a.\lambda x.(x = a)} = \lambda x.\false \qquad\text{and}\qquad \norm{\nu a.\nu b.\lambda x.(a\,b)x} = \lambda x.x.
  \end{equation*}
  The choice of $u = \Leak(V, s)$ in the expression case of our construction is crucial here; it is of course true that
  \begin{equation*}
    \lambda x.(a\,b)x \,(\id_{\{a, b\}})_{\tnames \to \tnames}\, \lambda x.(a\,b)x,
  \end{equation*}
  but this does not help us identify and eliminate the private names $a, b$.
\end{example}

\begin{proposition}\label{prop:normal-form-facts}
Let $\tau$ be a first-order type and $M \in \Safe^s_\tau$. Then
\begin{enumerate}
\item if $M$ is a value, so is $\norm{M, s}$;\label{prop:normal-form-facts-canonical}
\item the names that appear in $\norm{M, s}$ are a subset of the names that appear in $M$;\label{prop:normal-form-facts-name-inclusion}
\item $\norm{M, s}$ eliminates the names in $\Priv(M, s)$ (i.e. $\norm{M, s} \in \Exp_\tau(s)$);\label{prop:normal-form-facts-names}
\item $\norm{M, s} \in \Safe^s_\tau$;\label{prop:normal-form-facts-leak}
\item $\norm{M, s}$ is well-defined up to renaming bound variables and names; and\label{prop:normal-form-facts-wd}
\item $M \,(\id_s)_\tau\, \norm{M, s}$.\label{prop:normal-form-facts-relations}
\end{enumerate}
\end{proposition}

\begin{proof}
(\ref{prop:normal-form-facts-canonical}) is clear by construction. (\ref{prop:normal-form-facts-leak}) follows trivially from (\ref{prop:normal-form-facts-names}). We prove (\ref{prop:normal-form-facts-name-inclusion}), (\ref{prop:normal-form-facts-names}), (\ref{prop:normal-form-facts-wd}) and (\ref{prop:normal-form-facts-relations}) by induction on $\tau$, following the construction of the normal form $\norm{M, s}$. For (\ref{prop:normal-form-facts-name-inclusion}) and (\ref{prop:normal-form-facts-names}), the induction steps are clear and so is the case where $M$ is a value of type $\tbool$. If $M$ is a value of type $\tnames$, then $M \,(\id_s)_\tnames\, M$ implies that $M \in s$, and so $\norm{M, s} = M \in \Exp_\tnames(s)$. For (\ref{prop:normal-form-facts-wd}), the cases where $M$ is a value are clear, and the expression case follows because we made a canonical choice of $u = \Leak(V, s)$ in the construction of $\norm{M, s}$. For (\ref{prop:normal-form-facts-relations}), the expression case follows directly from the inductive hypothesis and the definition of logical relations. In the case where $M$ is a value and $\tau = \N \to \sigma$, we $\eta$-expand and write
\[M = \lambda x.\ite {x = n \in s \oplus t} {M_n} {M_0}.\]
We need to verify that $M_0 \,(\id_{s \oplus \{x\}})_\sigma\, \norm{M_0, s \oplus \{x\}}$ and that $M_n \,(\id_s)_\sigma\, \norm{M_n, s}$ for $n \in s$, both of which follow from the inductive hypothesis. The case that $M$ is a value and $\sigma = \tbool \to \tau$ is handled similarly.
\end{proof}

\begin{example}
As noted in \cref{eg:private-names}, $\Leak(\lambda x.(x = a)) = \{a\}$ and $\Leak(\lambda x.(a\,b)x) = \{a, b\}$, and these are indeed eliminated from the normal forms computed in \cref{eg:normal-forms}.
\end{example}

We can now equate the problem of checking if two terms are observationally equivalent to one of verifying the equality of their normal forms:

\begin{theorem}\label{thm:observational-equivalence-normal-form}
  Let $\sigma$ be a first-order type and let $M_i \in \Exp_\sigma(s \oplus t_i)$ for $i = 1, 2$. The following are equivalent:
  \begin{enumerate}
    \item $M_1 \,(\id_s)_\sigma\, M_2$;
    \item $M_i \in \Safe^s_\sigma$ and $\norm{M_1, s} = \norm{M_2, s}$ after possibly renaming bound variables and names.
  \end{enumerate}
\end{theorem}

\begin{proof}
If $M_i \in \Safe^s_\tau$ and $\norm{M_1, s} = \norm{M_2, s}$, then $\norm{M_1, s} \,(\id_s)_\sigma\, \norm{M_2, s}$ and so by transitivity of logical relations (\ref{lem:transitivity}) and \cref{prop:normal-form-facts} we have $M_1 \,(\id_s)_\sigma\, M_2$.

For the converse, suppose that $M_1 \,(\id_s)_\sigma\, M_2$. By transitivity, it is clear that $M_i \,(\id_s)_\sigma\, M_i$.

To show that $\norm{M_1, s} = \norm{M_2, s}$, we argue by induction, following the construction of the normal forms. The base case is clear. Now consider the inductive step at values. In the case that $\sigma = \tnames \to \tau$, we $\eta$-expand and write
\[M_i = \lambda x.\ite {x = n \in s \oplus t_i} {M^i_n} {M^i_0}.\]
By definition of logical relations, because $M_1 \,(\id_s)_\sigma\, M_2$, we have $M^1_0 \,(\id_{s \oplus \{x\}})_\sigma\, M^2_0$ and $M^1_n \,(\id_s)_\sigma\, M^2_n$ for $n \in s$. By our inductive hypothesis, this means that $\norm{M^1_0, s \oplus \{x\}} = \norm{M^2_0, s \oplus \{x\}}$ and $\norm{M^1_n, s} = \norm{M^2_n, s}$ for $n \in s$. It follows that $\norm{M_1, s} = \norm{M_2, s}$. The case that $\sigma = \tbool \to \tau$ is the same.

In the case of expressions, let $V_i \in \Val_\sigma(s \oplus t_i \oplus t'_i)$ be the values such that $s \oplus t_i \vdash M_i \Downarrow (t'_i)V_i$. Let $u_i = \Leak(V, s) \subseteq t_i'$. Then $V_i \in \Safe^{s \oplus u_i}_\tau$ and $\norm{M_i, s} = \nu u_i.\norm{V_i, s \oplus u_i}$. We know that $M_1 \,(\id_s)_\sigma\, M_2$, so there is some $R\colon t'_1 \leftrightharpoons t'_2$ such that $V_1 \,(\id_s \oplus R)_\sigma\, V_2$. By \cref{prop:logical-relations-minimality}, after possibly renaming names we have $u_1 = u_2 = u$, $\id_u \subseteq R$ and $V_1 \,(\id_{s \oplus u})_\sigma\, V_2$. We therefore have $\norm{V_1, s \oplus u} = \norm{V_2, s \oplus u}$ by our inductive hypothesis, and so $\norm{M_1, s} = \norm{M_2, s}$.
\end{proof}

\subsection{Full Abstraction at First-Order Types}
\label{sec:proof_fullabstraction}

At first-order types, it is sufficient to eliminate private names in order to prove abstraction:

\begin{theorem}\label{thm:abstract-iff-normal-form-equal}
  Let $\cat{C}$ be a categorical model of the $\nu$-calculus. $\cat{C}$ is fully abstract at first-order types if and only if for all first-order types $\tau$ and all $M \in \Exp_\tau(s)$, we have
  \begin{equation}\label{eq:cat-preserve-normal-form}
    \sem{M}_{\neq s} = \sem{\norm{M, s}}_{\neq s}.
  \end{equation}
\end{theorem}

\begin{proof}
That this is necessary is clear, as by \cref{prop:normal-form-facts} normal forms preserve logical relations and therefore (by \cref{thm:logical-relation-observational-equivalence}) observational equivalence. To see that it is sufficient, suppose that $\cat{C}$ satisfies \eqref{eq:cat-preserve-normal-form} and let $M_1, M_2 \in \Exp_\tau(s)$ for a first-order type $\tau$. If $M_1 \approx_\tau M_2$, then by \cref{thm:logical-relation-observational-equivalence,thm:observational-equivalence-normal-form} $\norm{M_1, s} = \norm{M_2, s}$, and so
\[\sem{M_1}_{\neq s} = \sem{\norm{M_1, s}}_{\neq s} = \sem{\norm{M_2, s}}_{\neq s} = \sem{M_2}_{\neq s}.\vspace{-6mm}\]
\end{proof}

For the remainder of this section, we will let the space of names be the circle $\T = [0, 1)$ and we will let $\nu$ be the uniform measure on $\T$ (we may assume this is the case by \cref{rmk:any-sbs}). We choose to work with the circle as there is a canonical group structure $(\T, +)$ on $\T$, namely addition modulo 1, that is both compatible with the measurable structure (and hence, by Prop.~\ref{prop:conservativity}, the quasi-Borel structure) of $\T$ and is \emph{$\nu$-invariant}. This means that the maps $+\colon \T \times \T \to \T$ and $-\colon \T \times \T \to \T$ are quasi-Borel, and for all $g \in \T$ and $B \subseteq \T$ Borel we have $\nu(g + B) = \nu(B)$. More generally, this implies that for all $f: \T \to P(X)$ and $g \in \T$, we have
\begin{equation*}
  \text{let $x \leftarrow \nu$ in $f(g + x)$} = \int_\T f(g + x)d\nu(x) = \int_\T f(x)d\nu(x) = \text{let $x \leftarrow \nu$ in $f(x)$}.
\end{equation*}
The idea of $\nu$-invariance will be used to treat private names as interchangeable in $\qbs$.

We will now use the $\nu$-invariant group structure on $\T$, along with privacy, to prove that passing to normal forms preserves $\qbs$ semantics.

\begin{example}\label{eg:trans-abstract}
  Consider the transposition $\nu a.\nu b.\lambda x.(a\,b)x$. We have seen that
  \begin{equation*}
  \norm{\nu a.\nu b.\lambda x.(a\,b)x} = \lambda x.x.
  \end{equation*}
  We wish to show that their semantics are equal in $\qbs$, i.e.~$\sem{\nu a.\nu b.\lambda x.(a\,b)x} = \sem{\lambda x.x} : P(P(\T)^\T).$
  To do this, we define a function $f: 2^\T \times \T^3 \to \T$ as follows:
  \begin{equation*}
    f(B, a, b, x) =
    \begin{cases}
      (x - a) + b & \mathsf{if}~ x - a \in B,\\
      (x - b) + a & \mathsf{else~if}~ x - b \in B,\\
      x & \mathsf{otherwise}.
    \end{cases}
  \end{equation*}
  This function behaves like a generalized transposition, parameterized by a new set-argument $B$. If $B = \emptyset$, then $f(\emptyset, a, b, x) = x$ is just the identity on $x$. If $B = \{g\}$ is a singleton, then
  \begin{equation*}
    f(\{g\}, a, b, x) =
    \begin{cases}
      g + b & \mathsf{if}~x = g + a,\\
      g + a & \mathsf{else~if}~ x = g + b,\\
      x & \mathsf{otherwise},
    \end{cases}
  \end{equation*}
  so that $f$ is a transposition whose parameters have been shifted by $g$.

  We then take $f': 2^\T \times \T^2 \to P(P(\T)^\T)$ to be the map $f'(B, a, b) = [\lambda x.[f(B, a, b, x)]]$, so that
  \begin{equation*}
    f'(\emptyset, a, b) = \sem{\lambda x.x} \quad\text{and}\quad f'(\{g\}, a, b) = \sem{\lambda x.(a\,b)x}(g + a, g + b),
  \end{equation*}
  and we define $h: 2^\T \to P(P(\T)^\T)$ to be
  \begin{equation*}
    h(B) = \letin a \nu \letin b \nu {f'(B, a, b)}.
  \end{equation*}
  It is clear that $h(\emptyset) = \sem{\lambda x.x}$. On the other hand, by the $\nu$-invariance of the action we have
  \begin{align*}
    h(\{g\}) &= \letin a \nu \letin b \nu {\sem{\lambda x.(a\,b)x}(g + a, g + b)}\\
    & = \letin a \nu \letin b \nu {\sem{\lambda x.(a\,b)x})(a, b)}\\
    & = \sem{\nu a.\nu b.\lambda x.(a\,b)x},
  \end{align*}
  independently of $g \in \T$.

  Our problem now reduces to the privacy equation. Specifically, we have
  \begin{align*}
    \sem{\lambda x.x} &= \letin B {[\emptyset]} {h(B)}\\
    & = \letin B {\left(\letin n \nu {[\{n\}]}\right)} {h(B)}\\
    & = \letin n \nu {h(\{n\})}\\
    & = \letin n \nu {\sem{\nu a.\nu b.\lambda x.(a\,b)x}}\\
    & = \sem{\nu a.\nu b.\lambda x.(a\,b)x},
  \end{align*}
  where the second equality is \eqref{eq:priv-c} and the final equality follows by discardability \eqref{eqn:affine}.
\end{example}

\begin{notation}
  If $\vec{t} = (t_1, \dots, t_n)$ is a vector in $\T^n$ and $g \in \T$, we write $g + \vec{t} = (g + t_1, \dots, g + t_n)$. Additionally, we write $\text{let $t \leftarrow \nu$}$ to be shorthand for drawing $t$ samples in a sequence:
  \begin{equation*}
    \letin {t_1} \nu \cdots \mathsf{let}~t_k \leftarrow \nu.
  \end{equation*}
\end{notation}

Now suppose that $\tau$ is a first-order type and $M \in \Exp_\tau(s)$. We will prove that $\sem{M} = \sem{\norm{M, s}}$ by constructing a function $f: 2^\T \times \T^{\neq s} \to P(\sem{\tau})$ satisfying
\[f(\emptyset, -) = \sem{\norm{M, s}}_{\neq s}(-) ~~\text{and}~~ f(\{n\}, -) = \sem{M}_{\neq s}(-),\]
as we did in \cref{eg:trans-abstract}, and applying the privacy equation \eqref{eq:priv-c}.

We will construct this $f$ inductively, parallel to the construction of the normal forms. In order to do this, we will provide a more general, parametrized version of this construction: given $M \in \Safe^s_\tau$ with names in $s \oplus t$, we will construct a function $f: 2^\T \times \T^{\neq s \oplus t} \to P(\sem{\tau})$ such that
\[f(\emptyset, -, \vec{t}) = \sem{\norm{M, s}}_{\neq s}(-) ~~\text{and}~~ f(\{n\}, -, \vec{t}) = \sem{M}_{\neq s}(-, n + \vec{t}).\]
We will use this parametrized version in the inductive step of our proof.

The construction of $f$ itself is somewhat high-level, and is analogous to the difference between the $\eta$-normal form of a term and its normal form. It takes as arguments a set of name-permutations $B$, a sequence $s$ of potentially leaked names, and a sequence of names $t$ that are guaranteed to remain private. It then identifies the redundant parts of the $\eta$-normal form --- where we compare against a private name $t_i$ --- and instead checks whether the name matches one of the names in $B + t_i$.

By selecting $B$ to be a fresh permutation $\vec{t} \leftrightarrow \vec{t'}$, we recover the semantics of the $\eta$-normal form. On the other hand, by letting $B$ be the empty set we skip redundant comparisons against private names, recovering the semantics of the normal form. We can then use the privacy equation to equate these two denotations, proving that each term is denotationally equivalent to its normal form.

\begin{proposition}\label{prop:reduction}
  Let $\tau$ be a first-order type and let $M \in \Exp_\tau(s \oplus t)$. If $M \in \Safe^s_\tau$, then there is a quasi-Borel map
  \begin{equation*}
    f: 2^\T \times \T^{\neq s \oplus t} \to P(\sem{\tau})
  \end{equation*}
  such that
  \begin{equation*}
    f(\emptyset, \vec{s}, \vec{t}) = \sem{\norm{M, s}}_{\neq s}(\vec{s}) ~~\text{and}~~ f(\{g\}, \vec{s}, \vec{t}) = \sem{M}_{\neq s \oplus t}(\vec{s}, g + \vec{t})
  \end{equation*}
  whenever $(\vec{s}, g + \vec{t}) \in \T^{\neq s \oplus t}$.

  In the case that $M = V$ is a value, $f$ factors through the unit of the monad.
\end{proposition}

\begin{proof}
We construct $f$ inductively, in parallel to the construction of the normal forms.

\textbf{Ground case:} If $\tau$ is a ground type and $V$ is a value, then $\norm{V, s} = V$ so we simply let
\begin{equation*}
  f(B, \vec{s}, \vec{t}) = \sem{\norm{V, s}}(\vec{s}) = [|\norm{V, s}|](\vec{s}).
\end{equation*}

\textbf{Function case $\tbool \to \tau$:} Suppose that $V$ is a value of type $\tbool \to \tau$ and that we have already constructed these functions for expressions of type $\tau$. We $\eta$-expand $V$, so that
\begin{equation*}
  V = \lambda x.\ite {x = \true} {M_\true} {M_\false}.
\end{equation*}
By definition of logical relations and the normal form we have $M_\true, M_\false \in \Safe^s_\tau$ and
\begin{equation*}
  \norm{V, s} = \lambda x.\ite {x = \true} {\norm{M_\true, s}} {\norm{M_\false, s}}.
\end{equation*}
By assumption we have functions $f_\true, f_\false: 2^\T \times \T^{\neq s \oplus t} \to P(\sem{\tau})$ satisfying the conditions of \cref{prop:reduction} for $M_\true$ and $M_\false$. We then define $f: 2^\T \times \T^{\neq s \oplus t} \to P(\sem{\tau})^B$ by
\begin{equation*}
  f(B, \vec{s}, \vec{t}) = \lambda x.
  \begin{cases}
    f_\true(B, \vec{s}, \vec{t}) & \mathsf{if}~ x = \true,\\
    f_\false(B, \vec{s}, \vec{t}) & \mathsf{otherwise.}
  \end{cases}
\end{equation*}
It is clear that $f(\emptyset, \vec{s}, \vec{t}) = |\norm{V, s}|(\vec{s})$ and $f(\{g\}, \vec{s}, \vec{t}) = |V|(\vec{s}, g + \vec{t})$ when $(\vec{s}, g + \vec{t}) \in \T^{\neq s \oplus t}$, so that $[f]$ satisfies \cref{prop:reduction} for $V$.

\textbf{Function case $\tnames \to \tau$:} Suppose that $V$ is a value of type $\tnames \to \tau$ and that we have already constructed these functions for expressions of type $\tau$. We $\eta$-expand $V$, so that
\begin{equation*}
  V = \lambda x.\ite {x = n \in s \oplus t} {M_n} {M_0}.
\end{equation*}
By definition of logical relations and the normal form we have $M_0 \in \Safe^{s \oplus \{x\}}_\tau$, $M_n \in \Safe^s_\tau$ for $n \in s$ and
\begin{equation*}
  \norm{V, s} = \lambda x.\ite {x = n \in s} {\norm{M_n, s}} {\norm{M_0, s \oplus \{x\}}}.
\end{equation*}
By assumption we have functions $f_n: 2^\T \times \T^{\neq s \oplus t} \to P(\sem{\tau})$ for $n \in s$ and $f_0: 2^\T \times \T^{\neq s \oplus t \oplus \{x\}} \to P(\sem{\tau})$ satisfying the conditions of \ref{prop:reduction} for $M_n$ and $M_0$. Writing $t = (t_1, \dots, t_k)$, we define $f: 2^\T \times \T^{\neq s \oplus t} \to P(\sem{\tau})^\T$ by
\begin{equation*}
  f(B, \vec{s}, \vec{t}) = \lambda x.
  \begin{cases}
    f_n(B, \vec{s}, \vec{t}) & \mathsf{if}~ x = n \in \vec{s},\\
    \sem{M_{t_1}}(\vec{s}, (x - t_1) + \vec{t}) & \mathsf{else~if}~ (x - t_1) \in B,\\
    \dots\\
    \sem{M_{t_k}}(\vec{s}, (x - t_k) + \vec{t}) & \mathsf{else~if}~ (x - t_k) \in B,\\
    f_0(B, \vec{s}, \vec{t}, x) & \mathsf{otherwise}.
  \end{cases}
\end{equation*}
If $B = \emptyset$, then
\begin{equation*}
  f(\emptyset, \vec{s}, \vec{t}) = \lambda x.
  \begin{cases}
    \sem{\norm{M_n, s}}(\vec{s}) & \mathsf{if}~ x = n \in \vec{s},\\
    \sem{\norm{M_0, s}}(\vec{s}) & \mathsf{otherwise}
  \end{cases}
\end{equation*}
so that $f(\emptyset, \vec{s}, \vec{t}) = |\norm{V, s}|(\vec{s})$. On the other hand, if $B = \{g\}$ is a singleton, then
\begin{equation*}
  f(\{g\}, \vec{s}, \vec{t}) = \lambda x.
  \begin{cases}
    \sem{M_n}(\vec{s}, g + \vec{t}) & \mathsf{if}~ x = n \in \vec{s},\\
    \sem{M_{t_1}}(\vec{s}, g + \vec{t}) & \mathsf{else~if}~ x = g + t_1,\\
    \dots\\
    \sem{M_{t_k}}(\vec{s}, g + \vec{t}) & \mathsf{else~if}~ x = g + t_k,\\
    \sem{M_0}(\vec{s}, g + \vec{t}, x) & \mathsf{otherwise}
  \end{cases}
\end{equation*}
so that $f(\{g\}, \vec{s}, \vec{t}) = |V|(\vec{s}, g + \vec{t})$ when $(\vec{s}, g + \vec{t}) \in \T^{\neq s \oplus t}$. Thus $[f]$ satisfies \cref{prop:reduction} for $V$.

\textbf{Expression case:} Suppose that we have constructed these reductions for values of type $\tau$. We have $M \,(\id_s)_\tau\, M$, so by definition of logical relations and the normal form there is some $V \in \Val_\tau(s \oplus t \oplus u \oplus w)$ such that $s \oplus t \vdash M \Downarrow_\tau (u \oplus w)V$ and $u = \Leak(V, s)$. Therefore $V \in \Safe^{s \oplus u}_\tau$ and
\begin{equation*}
  \norm{M, s} = \nu u.\norm{V, s \oplus u}.
\end{equation*}
By assumption, there is a function $f_V: 2^\T \times \T^{\neq s \oplus t \oplus u \oplus w} \to P(\sem{\tau})$ satisfying the conditions of \cref{prop:reduction} for $V$ and $\norm{V, s \oplus u}$. We then define $f: 2^\T \times \T^{\neq s \oplus t} \to P(\sem{\tau})$ by
\begin{equation*}
  f(B, \vec{s}, \vec{t}) = \letin u \nu \letin w \nu {f_V(B, \vec{s}, \vec{t}, \vec{u}, \vec{w})}.
\end{equation*}
It follows that
\begin{align*}
  f(\{g\}, \vec{s}, \vec{t}) &= \letin u \nu \letin w \nu {f_V(\{g\}, \vec{s}, \vec{t}, \vec{u}, \vec{w})}\\
  & = \letin u \nu \letin w \nu {\sem{V}_{\neq s \oplus t \oplus u \oplus w}(\vec{s}, g + \vec{t}, \vec{u}, g + \vec{w})}\\
  & = \letin u \nu \letin w \nu {\sem{V}_{\neq s \oplus t \oplus u \oplus w}(\vec{s}, g + \vec{t}, \vec{u}, \vec{w})}\\
  & = \sem{\nu u.\nu w.V}_{\neq s \oplus t}(\vec{s}, g + \vec{t})\\
  & = \sem{M}_{\neq s \oplus t}(\vec{s}, g + \vec{t})
\end{align*}
whenever $(\vec{s}, g + \vec{t}) \in \R^{\neq s \oplus t}$, where the third equality follows by $\nu$-invariance and the last by soundness (\cref{thm:sound-adequate}). Similarly, we verify that $f(\emptyset, \vec{s}, \vec{t}) = \sem{\norm{M, s}}_{\neq s}(\vec{s})$ by discardability \eqref{eqn:affine} and soundness.
\end{proof}

We note that this construction is not specific to quasi-Borel spaces; it can be performed completely syntactically in a metalanguage asserting that $N$ carries a $\nu$-invariant group structure.

It follows immediately that passing to normal forms preserves $\qbs$ semantics, and therefore that $\qbs$ is fully abstract at first-order types:

\begin{theorem}\label{thm:fullabstraction}
  $\qbs$ is fully abstract at first-order types.
\end{theorem}

\begin{proof}
  By \cref{thm:abstract-iff-normal-form-equal} it is enough to show that $\qbs$ validates passing to normal forms. Let $\tau$ be a first-order type and let $M \in \Exp_\tau(s)$. By \cref{prop:reduction} there is a quasi-Borel map $f: 2^\T \times \T^{\neq s} \to P(\sem{\tau})$ such that
  \begin{equation*}
    f(\emptyset, \vec{s}) = \sem{\norm{M, s}}_{\neq s}(\vec{s}) \quad\text{and}\quad f(\{g\}, \vec{s}) = \sem{M}_{\neq s}(\vec{s}).
  \end{equation*}
  Currying, we get a map $h: 2^\T \to P(\sem{\tau})^{\T^{\neq s}}$ such that
  \begin{equation*}
    h(\emptyset) = \sem{\norm{M, s}}_{\neq s} \quad\text{and}\quad h(\{n\}) = \sem{M}_{\neq s}.
  \end{equation*}
  It follows that
  \begin{align*}
    \sem{\norm{M, s}}_{\neq s} &= \letin B {[\emptyset]} {h(B)} = \letin B {\left(\letin n \nu {[\{n\}]}\right)} h(B)\\
    & = \letin n \nu {h(\{n\})} = \letin n \nu {\sem{M}_{\neq s}} = \sem{M}_{\neq s},
  \end{align*}
  where the second equality is \eqref{eq:priv-c} and the final equality follows by discardability \eqref{eqn:affine}.
\end{proof}

\section{Structural Consequences}
\label{sec:structural}

In this section, we highlight some consequences our main result has on the category of quasi-Borel spaces and other models of name generation. The privacy equation makes it impossible in $\catname{Qbs}$ to find certain conditional probabilities, as this would require revealing a private name (Prop.~\ref{prop:noconditionals}). This means care is needed for Bayesian inference in a higher-typed situation. We will give a broader context for this result using recent notions from synthetic probability theory, allowing us to consider any model of name generation as a categorical model of probability.

\begin{definition}[{\cite[11.1]{fritz}}]
  Let $\mu \in P(X \times Y)$ be a probability distribution and $\mu_X \in P(X)$ its first marginal. A \emph{conditional distribution} for $\mu$ is a morphism $\mu_{|X} : X \to P(Y)$ such that
  \[\mu = \letin x {\mu_X} {\letin y {\mu_{|X}(x)} {[(x,y)]}}.\]
\end{definition}

We will now consider the distribution $\mu \in P(2^\R \times \R)$
\begin{equation}
  \mu = \letin a \nu {[(\{a\}, a)]} \label{eq:leaking}
\end{equation}
which returns a closure with private name $a$, but also leaks the name $a$ in the second component.

\begin{proposition}\label{prop:noconditionals}
In $\catname{Qbs}$, conditionals need not exist at function types.
\end{proposition}

\begin{proof}
By the privacy equation (\ref{eq:priv-c}), the first marginal of $\mu$ (\ref{eq:leaking}) equals
\[ \mu_1 = \letin a \nu {[\{a\}]} = [\emptyset] \quad:P(2^\R). \]
If $\mu$ admitted a conditional distribution $\mu_{|1} : 2^\R \to P(\R)$, we would obtain
\[ \mu = \letin A {[\emptyset]} {\letin b {\mu_{|1}(A)} {[(A,b)]}} = \letin b {\mu_{|1}(\emptyset)} {[(\emptyset,b)]} \quad:P(2^\R\times \R)\text. \]
This is a contradiction, as the predicate $(\ni):2^\R\times \R\to 2$ is always true for $\mu$, and always false for the RHS. To condition on $\mu_1$ would mean to reconstruct the value $a$ given only access to the marginal $\{a\}$, which is impossible.
\end{proof}

All conditionals from practical statistics (at ground types like $\R$) are still supported by quasi-Borel spaces. The situation is different at function types, but this is not a coincidental pathology of $\catname{Qbs}$: Name generation offers a systematic reason why conditioning on function types is inconsistent. To make this precise, we will consider any model of name generation as a categorical model of probability theory, and study conditioning in that context. We show that the privacy equation is inconsistent with an axiom called `positivity', which is valid in traditional measure-theoretic probability, but not in $\catname{Qbs}$ by our full-abstraction result.

Categorical or synthetic probability theory is the abstract axiomatization of probabilistic systems. Its high-level nature ties it closely to the semantics of probabilistic programming languages: One could argue that such languages are precisely the internal languages of synthetic probability theories, and different axioms appear as admissible program equations (see \eqref{eq:det}). The subject has been explored among others by \cite{kock, scibior, fritz}. Of these approaches, we adopt the language of Markov categories which is increasingly widely used~\cite{quantum-markov,fritz,shiebler-likelihood,patterson-thesis}.

\begin{definition}[{\cite[2.1]{fritz}}]
  A \emph{Markov category} $\cat C$ is a symmetric monoidal category in which every object $X$ is equipped with the structure of a commutative comonoid $\cpy_X : X \to X \otimes X$, $\del_X : X \to I$ satisfying naturality conditions.
\end{definition}

Morphisms in a Markov category capture stochastic computation (Markov kernels); the interchange law of $\otimes$ encodes exchangeability/Fubini, and naturality of $\del$ the discardability of such computations. $\cpy$ allows us to introduce correlations. Morphisms $\mu : I \to X$ are called \emph{distributions} on $X$. Product distributions are formed by the tensor product, and if $\mu : I \to X \otimes Y$ is a distribution, we can take its marginals $\mu_X = (\id_X \otimes \del_Y) \circ \mu, \mu_Y= (\del_X \otimes \id_Y) \circ \mu$.

An important class of examples are Kleisli categories. If $T$ is a commutative and affine monad on a category $\cat C$ with finite products, then the Kleisli category $\mathbf{Kl}(T)$ is a Markov category \cite[3.2]{fritz}. Examples are the categories $\set$, $\meas$ and $\qbs$, all equipped with their respective probability monads. We observe that name generation (cf. Def.~\ref{def:catmodel}) is a synthetic probabilistic effect.

\begin{observation}\label{obs:name-gen-markov}
  For every categorical model $(\cat C, T)$ of the $\nu$-calculus, the category $\catname{Kl}(T)$ is a Markov category.
\end{observation}

\begin{proof}
  The monad $T$ is assumed commutative and affine, so we apply \cite[3.2]{fritz}.
\end{proof}

This makes the probabilistic semantics of this paper conceptually very natural: We have taken a synthetic probabilistic effect and given an interpretation using actual randomness. In what follows, we will explore some of the structural differences between name generation and traditional probability theory. By our full abstraction result, this behaviour will apply to quasi-Borel spaces as well.

We let $\cat C$ denote a Markov category and recall the following definitions

\begin{definition}[{\cite[10.1]{fritz}}]
  A morphism $f : X \to Y$ is \emph{deterministic} if it commutes with copying:
  \begin{equation*}
    \cpy_Y \circ f = (f \otimes f) \circ \cpy_X.
  \end{equation*}
\end{definition}

In the case of Kleisli categories, determinism is equivalent to the following program equation in the metalanguage:
\begin{equation}\label{eq:det}
x : X \vdash \letin y {f(x)} [(y,y)] = \letin {y_1} {f(x)} {\letin {y_2} {f(x)} {[(y_1,y_2)]}} : T(Y \times Y)
\end{equation}
Note that any morphism that factors through the unit of the monad is deterministic, but the converse is false in general.

\begin{definition}[{\cite[11.22]{fritz}}]\label{def:positive}
  A Markov category $\cat C$ is called \emph{positive} if whenever $f : X \to Y$ and $g : Y \to Z$ are such that $g \circ f$ is deterministic, then
  \begin{equation*}
    (g \otimes \id_Y) \circ \cpy_Y \circ f = ((g \circ f) \otimes f) \circ \cpy_X.
  \end{equation*}
\end{definition}

This equation is valid in discrete and measure-theoretic probability by \cite[11.25]{fritz}. We suggest the reading that ``irrelevant intermediate results cannot introduce correlations'': On the RHS, the output of $f$ is resampled instead of copied. This blatantly fails in the presence of negative probabilities: There is a monad $D_\pm$ on $\set$ assigning to $X$ distributions which sum to $1$, but whose weights can be negative. Probabilities thus are allowed to interfere destructively. The Kleisli category of $D_\pm$ is still a valid Markov category, and it is in this positivity axiom that its theory deviates from standard probability \cite[11.27]{fritz}. A consequence of positivity is this:

\begin{proposition}[One deterministic marginal]\label{prop:one-deterministic-marginal}
  Let $\cat C$ be a positive Markov category, and ${\mu \colon I \to X \otimes Y}$ be a distribution. If the marginal $\mu_X : I \to X$ is deterministic, then $\mu = \mu_X \otimes \mu_Y$.
\end{proposition}

\begin{proof}
  Let $f=\mu$ and $g : X \otimes Y \to X$ be marginalization. By assumption $g \circ f$ is deterministic. The result is obtained by simple string diagram manipulation from the positivity axiom.
\end{proof}

In $\catname{Meas}$, nothing can be correlated with a constant: If $(X,Y)$ is a joint distribution and $X \overset{d}{=} x_0$ is deterministic, then $Y$ is independent from $X$. The privacy equation implies that this does not hold for name generation, analogously to Prop. \ref{prop:noconditionals}.

\begin{proposition}\label{prop:nu-nonpositive}
  Any non-degenerate model of the $\nu$-calculus that verifies \eqref{eq:priv-c} is non-positive.
\end{proposition}

\begin{proof}
  Consider the distribution $\mu = \letin a \new {[(\{a\}, a)]}$. Its first marginal is deterministic, as $\mu_1 = {\letin a \new {[\{a\}]}} = [\emptyset]$ by \eqref{eq:priv-c}. Yet $\mu$ is not the product of its marginals $[\emptyset] \otimes \new$, as the map $(\ni) : B^N \times N \to B$ distinguishes the two distributions. This violates Prop.~\ref{prop:one-deterministic-marginal}.
\end{proof}

\begin{corollary}\label{prop:qbs-nonpositive}
  The category $\qbs$ is not positive at function spaces.
\end{corollary}

We have thus given a natural example of a non-positive Markov category, and this phenomenon has an intuitive meaning in the context of name generation. Any fixed singleton set $\{a\}$ is manifestly distinguishable from $\emptyset$, but only if we know where to look. By randomizing $a$, its value is perfectly anonymized and this information is lost, leaving us with the empty set. This is reminiscent of a limited form of destructive interference. Note that probabilities in quasi-Borel spaces remain non-negative.

The concept of non-positivity is useful to connect several structural properties of $\catname{Qbs}$. Firstly, it explains the non-existence of conditionals and disintegrations in Prop \ref{prop:noconditionals}, as by \cite[11.24]{fritz} conditionals imply positivity. Secondly, the failure of the functor $\Sigma : \qbs \to \meas$ (Prop.~\ref{prop:qbs-meas-adjunction}) to preserve products is necessary in order to violate Proposition \ref{prop:one-deterministic-marginal}, as we observe

\begin{observation}\label{obs:product-of-marginals}
  Let $X,Y$ be quasi-Borel spaces and $\mu \in P(X \times Y)$ such that $\mu_X = [x]$ for some $x \in X$. If $\Sigma(X \times Y) \cong \Sigma X \times \Sigma Y$, then $\mu$ is the product of its marginals.
\end{observation}

\begin{proof}
  If $X \times Y$ carries a product-$\sigma$-algebra, the situation reduces to $\catname{Meas}$, which is positive.
\end{proof}

Proposition \ref{prop:nu-nonpositive} thus implies that the product $2^\R \times \R$ cannot be preserved. Similar arguments can be constructed for other product spaces like $2^\R \times 2^\R$. Another structural result on quasi-Borel spaces that follows from the methods of \S\ref{sec:fullabstraction} concerns the novel status of function spaces.

\begin{proposition}\label{thm:2r-not-measurable}
  The quasi-Borel space $2^\R$ is not isomorphic to $M(\Omega)$ for any measurable space $\Omega$.
\end{proposition}

\begin{proof}
  The adjunction $\Sigma \dashv M$ (Prop.~\ref{prop:qbs-meas-adjunction}) is idempotent, hence a quasi-Borel space $X$ lies in the essential image of $M$ if and only if $M_X = M_{\Sigma_X}$. We will show that $M_{2^\R}$ is strictly smaller than $M_{\Sigma_{2^\R}}$. Let $f : \R \to \R$ be a bijective function that is not measurable, and let $A \subseteq \mathbb R^2$ be the graph of $f$. By \cite[Theorem 4.5.2]{srivastava}, $A$ is not Borel and hence the map $\alpha : \R \to 2^\R, x \mapsto A_x = \{f(x)\}$ does not lie in $M_{2^\R}$. However $\alpha \in M_{\Sigma_{2^\R}}$, that is $\alpha$ is a measurable map from $\R$ to $(2^\R, \Sigma_{2^\R})$. Namely, for every $\mathcal U \in \Sigma_{2^\R}$, we have $\alpha^{-1}(\mathcal U) = \{ x : \{f(x)\} \in \mathcal U \}$. By Lemma~\ref{lemma:bob-singleton}, the set $S = \{ x : \{x\} \in \mathcal U \}$ is always countable or cocountable, and so is $\alpha^{-1}(\mathcal U) = f^{-1}(S)$ by bijectivity of $f$. So the preimage is a Borel set as desired.
\end{proof}

\section{Related Work and Context}
\label{sec:relwork}

\subsection{Names in Computer Science and Statistics}
\label{sec:related-names}

Names are important in almost every area of practical computer science. There are two main ways to implement name generation: the first is to have one or more servers that deterministically supply fresh names as requested, and the second is to pick them randomly. This paper has emphasised the surprising effectiveness of the latter approach for programming semantics, in that it provides a model that is fully abstract up to first order, not by construction, but by general properties of the real numbers.

Names might be server names in distributed systems, nonces in cryptography, object names in object oriented programming, gensym in Lisp, or abstract memory locations in heap-based programming. Beyond computer science, names play a vital role in logic and set theory.
Since this paper is in the theme of probabilistic programming, we emphasise in particular two ways that names are used in probabilistic programming and statistics, and the way that name generation is already understood in terms of randomness there.
\begin{itemize}
  \item The Dirichlet process can be used as a method for clustering data points where the number of clusters is unknown. The `base distribution' of a Dirichlet process allocates a label to each cluster that is discovered. It is common to use an atomless distribution such as a Gaussian for this, so that the labels are in effect fresh names for the clusters. In the Church probabilistic programming language, it is common to actually use Lisp's gensym as the base distribution for the Dirichlet process~\cite{roy-gensym}.
  \item A graphon is a measurable function $g\colon [0,1]^2\to [0,1]$, and determines a countably infinite random graph in the following way: we label nodes in the graph with numbers drawn uniformly from $[0,1]$, and there is an edge between two nodes $r,s$ with probability $g(r,s)$. Thus when building a graph node-by-node, the name of each fresh node is, in effect, a real number~\cite{orbanz-roy}.
\end{itemize}

While many programming languages support name generation directly or through libraries, we have here focussed on the $\nu$-calculus, which is stripped down so that the relationship between name generation and functions can be investigated. There are many other calculi for names, including $\lambda\nu$, which is a call-by-name analogue of the $\nu$-calculus~\cite{odersky:localnames}, and the $\pi$-calculus, for concurrency~\cite{milner-pi}.
Moreover, research on the  $\nu$-calculus has led to significant developments in different directions,
including memory references (e.g.~\cite{jeffrey-rathke,laird-games,murawski-tzevelekos}) and
cryptographic protocols (e.g.~\cite{sumii-pierce}).
It may well be informative to pursue quasi-Borel based analyses of these applications in the future.

\subsection{Models of the \texorpdfstring{$\nu$}{Nu}-Calculus}
\label{sec:nu}

Arguably the simplest model of the $\nu$-calculus is a set-theoretic model with a special set~$N$ of atoms, where abstractness of the atoms is enforced by an invariance property under permutations of the atoms. This model appears in different equivalent guises, including nominal sets and sheaves on finite sets of names and injective renamings. In this model, types are interpreted as sets, and expressions are interpreted as equivariant functions; see for instance \cite[Ch.~9]{pitts:nominalsets} or \cite[\S 3.7]{stark:thesis}. In nominal sets, equivariance is used to treat private names as interchangeable, which is reminiscent of the idea of $\nu$-invariance in \ref{sec:proof_fullabstraction}.

This simple model of nominal sets is very useful, but on its own it is only fully abstract at ground types~\cite[\S5]{stark:cmln}.
The privacy law~(\ref{eq:priv-c}) fails because the Boolean existence function ${\exists:(N\to B)\to B}$ (\ref{eqn:exists}) is a morphism of nominal sets, and so we can distinguish the expressions in (\ref{eq:priv-c}) via the context
\begin{equation}\label{eqn:ctx-ex}\letin f{(-)}(\exists f):B\text.\end{equation}
Nominal sets are a Boolean model of set theory~\cite[Thm.~2.23]{pitts:nominalsets}, and one would necessarily have this kind of existence function $\exists$ in any Boolean model of set theory.
Quasi-Borel spaces do form a kind-of model of set theory (a quasitopos), but it is an intuitionistic one, and there is no Boolean existence function (\cref{ex:qbs-nonemptiness}).

To deal with this incompleteness of nominal sets, Stark~\cite[\S 4.4]{stark:thesis} proposed a semantic version of the logical relations that we have recalled in Section~\ref{sec:fullabstraction}. This model, based on functors between double categories, is fully abstract at first order, as ours is.
Subsequently an alternative logical relations model was proposed by~\cite{zhang-nowak}, by working with logical relations over a functor category that more clearly distinguishes between public and private names.
$\qbs$ is different in spirit to these models, as it is a general purpose model of probability theory rather than a model purpose-built for full abstraction. A quasi-Borel space can be regarded as an $\R$-indexed logical relation (in the sense of~\cite{plotkin-lambda-def}), but it also has a basic role motivated by probability theory.

One curious aspect is that all of these models of the $\nu$-calculus will provide unusual Markov categories (Observation~\ref{obs:name-gen-markov}), i.e.~categorical models of probability theory, even if they do not exhibit any randomness in the familiar sense.
\paragraph{Full Abstraction at Higher Types.}
None of the set-based models justify the following observational equivalence at second-order \cite[Ex.~4(3)]{stark:whatsnew}:
\begin{equation}\label{eqn:tricky}
  \nu a. \nu b. \lambda f. (f a \Leftrightarrow f b) \approx_{(\tnames \to \tbool) \to \tbool} \lambda f. \true
\end{equation}
where $\Leftrightarrow$ denotes the biconditional of booleans.
To see that this equation fails in the quasi-Borel space model, notice that there is a $\qbs$ morphism $(0{>})\colon \R\to 2$ given by $(0{>})(r)=\true$ iff $0>r$, and so we can temporarily add this as a constant to the $\nu$-calculus and keep the rest of the denotational semantics the same.
Then $\sem{(\lambda f.\true)(0{>})}=\sem{\true}$, but $\sem{(\nu a.\nu b. \lambda f(f a\Leftrightarrow f b))(0{>})}$ is different; informally it returns $\true$ with probability $0.5$.

To our knowledge, the only models of~(\ref{eqn:tricky}) to date are game-semantic models~\cite{abramsky+:nominalgames,tzevelekos-thesis} and bisimulation models~\cite{benton-koutavas}. In common with our work, normal forms play an implicit role in those models, but those models are very different from ours at higher types.
In the future it may be interesting to impose further invariance properties on quasi-Borel spaces to bridge the gap.

\paragraph{Usage of Models in Practice.}
One major application of models is in validating observational equivalences that may be used for compiler optimizations. In probabilistic programming, optimizations are performed as part of statistical inference algorithms. For instance, discardability~(\ref{eqn:affine}) and exchangeability~(\ref{eqn:commutativity}) are simple but useful translations in practical inference~\cite{birch,r2}, and partial evaluation and normalization are used in several systems~\cite{lambda-psi,shan-ramsey}.
Our work in this paper is primarily foundational, but one application is that, in a higher-order probabilistic language, a statistical inference algorithm could legitimately simplify using our normalization algorithm (\S\ref{sec:nf}) or higher-typed equations such as the privacy equation~(\ref{eq:privnu}).

\subsection{Other Models of Higher-Order Probability}
\label{sec:related-hoprob}

In this paper we have focused on quasi-Borel spaces, but recently other models of higher-order probability have been proposed. We contend that there are two essential ingredients for using a model of higher-order probability to model the $\nu$-calculus, with name generation as randomness:
\begin{enumerate}
  \item it must support an atomless distribution, such as the normal distribution, on some uncountable space $N$;
  \item it must support equality checking on that space, as a function $N\times N\to 2$.
\end{enumerate}
Some models, such as probabilistic coherence spaces~\cite{pcoh}, do not seem to support atomless distributions, which makes it unclear how to use them for this purpose. Other models are based on the idea that all functions are continuous or computable, e.g.~\cite{escardo-testing,computable-probprog} and then it is impossible to have equality checking for $N=\mathbb{R}$.

This still leaves several recent models, including the stable cones model~\cite{ehrhard:cones}, a function analytic model~\cite{dahlqvist:roban}, game semantics~\cite{paquet-games}, geometry of interaction~\cite{goi-bayes}, boolean-valued sets~\cite{scott-bv-prob}, a boolean topos model~\cite{simpson-sheaves}, and an operational bisimulation~\cite{dal-lago}. There are also recent logics for higher order probability~\cite{sato:formal-verification}. We understand from the authors that operational bisimulation violates the privacy law, for an interesting reason, and that the boolean topos model violates it because of booleanness (as above, (\ref{eqn:ctx-ex})). It remains to be seen how abstract the other recent models are for interpreting the $\nu$-calculus. We note that~\cite{dahlqvist:roban,ehrhard:cones} are currently focused on call-by-\emph{name} semantics and so it is not obvious how to use them with the call-by-\emph{value} $\nu$-calculus that we considered in this paper (see~\eqref{eg:nu-lambda-commute}).

Finally we mention another model of higher-order probability that is purely combinatorial~\cite{staton:betabernoulli}. That work emphasizes two views of the same model. From one point of view, the space $N$ is a space of real numbers and supports the beta distributions (which are atomless). From another point of view, $N$ is a space of freshly generated names of urns, and real numbers do not arise. This is not a model of the $\nu$-calculus since it does not support name equality checking, but it is related in spirit nonetheless.

\subsection{Beyond \texorpdfstring{$\nu$}{Nu}-Calculus}

The $\nu$-calculus describes the basic interaction between functions and name generation. Going further, it is also important to investigate the situation where the names have further meaning or structure. In probabilistic programming and statistics, the reorderability of names amounts to sequence exchangeability (e.g.~\cite{staton:betabernoulli}), and this is of fundamental importance in statistics and probabilistic programming. But more elaborate symmetries and exchangeabilities  are also important (e.g.~\cite{orbanz-roy,hier-exch,syafr-pps}), and we leave this for future work.

\begin{acks}

We thank Alexander Kechris for a first proof of the privacy equation; we have independently developed a different proof based on Borel inseparability (\S\ref{sec:privacy}). We also thank Ohad Kammar for many insightful comments on an early draft of this work. The work has had three starting points: one in discussions with Alex Simpson in 2013; one in discussions with Cameron Freer and Dan Roy in 2016; and the last following discussions with Ohad Kammar and Prakash Panangaden in 2019. We also thank Tobias Fritz, Mathieu Huot and Sean Moss for helpful discussions. It has been helpful to present preliminary versions of this work at the LAFI and PPS workshops. This work is supported by
\grantsponsor{EPSRC}{EPSRC}{https://epsrc.ukri.org/} Grant No.~\grantnum{EPSRC}{EP/N509711/1},
a \grantsponsor{RS}{Royal Society University Research Fellowship}{http://www.royalsociety.org},
\grantsponsor{FRQNT-bourse-maitrise}{FRQNT}{http://www.frqnt.gouv.qc.ca/} Grant No.~\grantnum{FRQNT-bourse-maitrise}{290736},
\grantsponsor{NSERC}{NSERC Discovery Grant}{https://www.nserc-crsng.gc.ca/} No.~\grantnum{NSERC}{RGPIN-2020-05445},
\grantsponsor{NSERC-supplement}{NSERC Discovery Accelerator Supplement}{} No.~\grantnum{NSERC-supplement}{RGPAS-2020-00097}
and \grantsponsor{NCN}{NCN Grant Harmonia}{https://ncn.gov.pl/} No.~\grantnum{NCN}{2018/30/M/ST1/00668}.

\end{acks}

\bibliography{paper}


\begin{thebibliography}{58}


\ifx \showCODEN    \undefined \def \showCODEN     #1{\unskip}     \fi
\ifx \showDOI      \undefined \def \showDOI       #1{#1}\fi
\ifx \showISBNx    \undefined \def \showISBNx     #1{\unskip}     \fi
\ifx \showISBNxiii \undefined \def \showISBNxiii  #1{\unskip}     \fi
\ifx \showISSN     \undefined \def \showISSN      #1{\unskip}     \fi
\ifx \showLCCN     \undefined \def \showLCCN      #1{\unskip}     \fi
\ifx \shownote     \undefined \def \shownote      #1{#1}          \fi
\ifx \showarticletitle \undefined \def \showarticletitle #1{#1}   \fi
\ifx \showURL      \undefined \def \showURL       {\relax}        \fi
\providecommand\bibfield[2]{#2}
\providecommand\bibinfo[2]{#2}
\providecommand\natexlab[1]{#1}
\providecommand\showeprint[2][]{arXiv:#2}

\bibitem[\protect\citeauthoryear{Abramsky, Ghica, Murawski, Ong, and
  Stark}{Abramsky et~al\mbox{.}}{2004}]%
        {abramsky+:nominalgames}
\bibfield{author}{\bibinfo{person}{S. Abramsky}, \bibinfo{person}{D.~R. Ghica},
  \bibinfo{person}{A.~S. Murawski}, \bibinfo{person}{C.-H.~L. Ong}, {and}
  \bibinfo{person}{I.~D.~B. Stark}.} \bibinfo{year}{2004}\natexlab{}.
\newblock \showarticletitle{Nominal games and full abstraction for the
  nu-Calculus}. In \bibinfo{booktitle}{\emph{Proc.~LICS 2004}}.
  \bibinfo{pages}{150 -- 159}.
\newblock
\showISBNx{0769521924}


\bibitem[\protect\citeauthoryear{Aumann}{Aumann}{1961}]%
        {aumann}
\bibfield{author}{\bibinfo{person}{Robert~J. Aumann}.}
  \bibinfo{year}{1961}\natexlab{}.
\newblock \showarticletitle{Borel structures for function spaces}.
\newblock \bibinfo{journal}{\emph{Illinois Journal of Mathematics}}
  \bibinfo{volume}{5} (\bibinfo{year}{1961}).
\newblock


\bibitem[\protect\citeauthoryear{Bacci, Furber, Kozen, Mardare, Panangaden, and
  Scott}{Bacci et~al\mbox{.}}{2018}]%
        {scott-bv-prob}
\bibfield{author}{\bibinfo{person}{Giorgio Bacci}, \bibinfo{person}{Robert
  Furber}, \bibinfo{person}{Dexter Kozen}, \bibinfo{person}{Radu Mardare},
  \bibinfo{person}{Prakash Panangaden}, {and} \bibinfo{person}{Dana Scott}.}
  \bibinfo{year}{2018}\natexlab{}.
\newblock \showarticletitle{Boolean-valued semantics for stochastic
  lambda-calculus}. In \bibinfo{booktitle}{\emph{Proc.~LICS 2018}}.
\newblock


\bibitem[\protect\citeauthoryear{Benton and Koutavas}{Benton and
  Koutavas}{2008}]%
        {benton-koutavas}
\bibfield{author}{\bibinfo{person}{Nick Benton} {and}
  \bibinfo{person}{Vasileios Koutavas}.} \bibinfo{year}{2008}\natexlab{}.
\newblock \bibinfo{booktitle}{\emph{A Mechanized Bisimulation for the
  Nu-Calculus}}.
\newblock \bibinfo{type}{{T}echnical {R}eport} MSR-TR-2008-129.
  \bibinfo{institution}{Microsoft Research}.
\newblock


\bibitem[\protect\citeauthoryear{chieh Shan and Ramsey}{chieh Shan and
  Ramsey}{2017}]%
        {shan-ramsey}
\bibfield{author}{\bibinfo{person}{Chung chieh Shan} {and}
  \bibinfo{person}{Norman Ramsey}.} \bibinfo{year}{2017}\natexlab{}.
\newblock \showarticletitle{Exact {B}ayesian inference by symbolic
  disintegration}. In \bibinfo{booktitle}{\emph{Proc.~POPL 2017}}.
\newblock


\bibitem[\protect\citeauthoryear{Dahlqvist and Kozen}{Dahlqvist and
  Kozen}{2020}]%
        {dahlqvist:roban}
\bibfield{author}{\bibinfo{person}{Fredrik Dahlqvist} {and}
  \bibinfo{person}{Dexter Kozen}.} \bibinfo{year}{2020}\natexlab{}.
\newblock \showarticletitle{Semantics of higher-order probabilistic programs
  with conditioning}.
\newblock \bibinfo{journal}{\emph{Proc. ACM Program. Lang.}}
  \bibinfo{volume}{4}, \bibinfo{number}{POPL}, Article \bibinfo{articleno}{19}
  (\bibinfo{date}{Dec.} \bibinfo{year}{2020}).
\newblock


\bibitem[\protect\citeauthoryear{{Dal~Lago} and Hoshino}{{Dal~Lago} and
  Hoshino}{2019}]%
        {goi-bayes}
\bibfield{author}{\bibinfo{person}{Ugo {Dal~Lago}} {and}
  \bibinfo{person}{Naohiko Hoshino}.} \bibinfo{year}{2019}\natexlab{}.
\newblock \showarticletitle{The geometry of {B}ayesian programming}. In
  \bibinfo{booktitle}{\emph{Proc.~LICS 2019}}.
\newblock


\bibitem[\protect\citeauthoryear{Ehrhard, Pagani, and Tasson}{Ehrhard
  et~al\mbox{.}}{2018}]%
        {ehrhard:cones}
\bibfield{author}{\bibinfo{person}{Thomas Ehrhard}, \bibinfo{person}{Michele
  Pagani}, {and} \bibinfo{person}{Christine Tasson}.}
  \bibinfo{year}{2018}\natexlab{}.
\newblock \showarticletitle{Measurable cones and stable, measurable functions}.
  In \bibinfo{booktitle}{\emph{Proc.~POPL 2018}}.
\newblock


\bibitem[\protect\citeauthoryear{Ehrhard, Tasson, and Pagani}{Ehrhard
  et~al\mbox{.}}{2014}]%
        {pcoh}
\bibfield{author}{\bibinfo{person}{Thomas Ehrhard}, \bibinfo{person}{Charistine
  Tasson}, {and} \bibinfo{person}{Michele Pagani}.}
  \bibinfo{year}{2014}\natexlab{}.
\newblock \showarticletitle{Probabilistic coherence spaces are fully abstract
  for probabilistic {P}{C}{F}}. In \bibinfo{booktitle}{\emph{Proc.~POPL 2014}}.
  \bibinfo{pages}{309--320}.
\newblock


\bibitem[\protect\citeauthoryear{Escardo}{Escardo}{2009}]%
        {escardo-testing}
\bibfield{author}{\bibinfo{person}{M.H. Escardo}.}
  \bibinfo{year}{2009}\natexlab{}.
\newblock \showarticletitle{Semi-decidability of may, must and probabilistic
  testing in a higher-type setting}. In \bibinfo{booktitle}{\emph{Proc.~MFPS
  2009}}.
\newblock


\bibitem[\protect\citeauthoryear{Fritz}{Fritz}{2020}]%
        {fritz}
\bibfield{author}{\bibinfo{person}{Tobias Fritz}.}
  \bibinfo{year}{2020}\natexlab{}.
\newblock \showarticletitle{A synthetic approach to Markov kernels, conditional
  independence and theorems on sufficient statistics}.
\newblock \bibinfo{journal}{\emph{Adv.~Math.}} \bibinfo{volume}{370},
  \bibinfo{number}{107239} (\bibinfo{date}{Aug.} \bibinfo{year}{2020}).
\newblock


\bibitem[\protect\citeauthoryear{Gehr, Steffen, and Vechev}{Gehr
  et~al\mbox{.}}{2020}]%
        {lambda-psi}
\bibfield{author}{\bibinfo{person}{T. Gehr}, \bibinfo{person}{S. Steffen},
  {and} \bibinfo{person}{M.~T. Vechev}.} \bibinfo{year}{2020}\natexlab{}.
\newblock \showarticletitle{$\lambda${P}{S}{I}: exact inference for
  higher-order probabilistic programs}. In \bibinfo{booktitle}{\emph{Proc.~PLDI
  2020}}.
\newblock


\bibitem[\protect\citeauthoryear{Giry}{Giry}{1982}]%
        {giry}
\bibfield{author}{\bibinfo{person}{Mich{\`e}le Giry}.}
  \bibinfo{year}{1982}\natexlab{}.
\newblock \showarticletitle{A categorical approach to probability theory}.
\newblock In \bibinfo{booktitle}{\emph{Categorical Aspects of Topology and
  Analysis}}. \bibinfo{series}{Lecture Notes in Mathematics},
  Vol.~\bibinfo{volume}{915}. \bibinfo{publisher}{Springer},
  \bibinfo{pages}{68--85}.
\newblock


\bibitem[\protect\citeauthoryear{Heunen, Kammar, Staton, and Yang}{Heunen
  et~al\mbox{.}}{2017}]%
        {heunen:qbs}
\bibfield{author}{\bibinfo{person}{Chris Heunen}, \bibinfo{person}{Ohad
  Kammar}, \bibinfo{person}{Sam Staton}, {and} \bibinfo{person}{Hongseok
  Yang}.} \bibinfo{year}{2017}\natexlab{}.
\newblock \showarticletitle{A Convenient Category for Higher-Order Probability
  Theory}. In \bibinfo{booktitle}{\emph{Proceedings of the 32nd Annual ACM/IEEE
  Symposium on Logic in Computer Science}} (Reykjav\'{\i}k, Iceland)
  \emph{(\bibinfo{series}{LICS '17})}. \bibinfo{publisher}{IEEE Press}, Article
  \bibinfo{articleno}{77}, \bibinfo{numpages}{12}~pages.
\newblock
\showISBNx{9781509030187}


\bibitem[\protect\citeauthoryear{Huang, Morrisett, and Spitters}{Huang
  et~al\mbox{.}}{2018}]%
        {computable-probprog}
\bibfield{author}{\bibinfo{person}{Daniel Huang}, \bibinfo{person}{Greg
  Morrisett}, {and} \bibinfo{person}{Bas Spitters}.}
  \bibinfo{year}{2018}\natexlab{}.
\newblock \bibinfo{title}{An application of computable distributions to the
  semantics of probabilistic programs}.
\newblock \bibinfo{howpublished}{arxiv:1806.07966}.
\newblock


\bibitem[\protect\citeauthoryear{Jeffrey and Rathke}{Jeffrey and
  Rathke}{1999}]%
        {jeffrey-rathke}
\bibfield{author}{\bibinfo{person}{A. Jeffrey} {and} \bibinfo{person}{J.
  Rathke}.} \bibinfo{year}{1999}\natexlab{}.
\newblock \showarticletitle{Towards a theory of bisimulation for local names}.
  In \bibinfo{booktitle}{\emph{Proc.~LICS 1999}}.
\newblock


\bibitem[\protect\citeauthoryear{Jung, Lee, Staton, and Yang}{Jung
  et~al\mbox{.}}{2020}]%
        {hier-exch}
\bibfield{author}{\bibinfo{person}{Paul Jung}, \bibinfo{person}{Jiho Lee},
  \bibinfo{person}{Sam Staton}, {and} \bibinfo{person}{Hongseok Yang}.}
  \bibinfo{year}{2020}\natexlab{}.
\newblock \showarticletitle{A generalization of hierarchical exchangeability on
  trees to directed acyclic graphs}.
\newblock \bibinfo{journal}{\emph{Annales Henri Lebesgue}}
  (\bibinfo{year}{2020}).
\newblock
\newblock
\shownote{to appear.}


\bibitem[\protect\citeauthoryear{Kallenberg}{Kallenberg}{2002}]%
        {kallenberg}
\bibfield{author}{\bibinfo{person}{Olav Kallenberg}.}
  \bibinfo{year}{2002}\natexlab{}.
\newblock \bibinfo{booktitle}{\emph{Foundations of Modern Probability}}.
\newblock \bibinfo{publisher}{Springer, New York}.
\newblock


\bibitem[\protect\citeauthoryear{Kammar and Plotkin}{Kammar and
  Plotkin}{2012}]%
        {Kammar-Plotkin}
\bibfield{author}{\bibinfo{person}{Ohad Kammar} {and}
  \bibinfo{person}{Gordon~D. Plotkin}.} \bibinfo{year}{2012}\natexlab{}.
\newblock \showarticletitle{Algebraic foundations for effect-dependent
  optimisations}. In \bibinfo{booktitle}{\emph{Proc.~POPL 2012}}.
  \bibinfo{pages}{349--360}.
\newblock


\bibitem[\protect\citeauthoryear{Kechris}{Kechris}{1987}]%
        {kechris}
\bibfield{author}{\bibinfo{person}{Alexander Kechris}.}
  \bibinfo{year}{1987}\natexlab{}.
\newblock \bibinfo{booktitle}{\emph{Classical Descriptive Set Theory}}.
\newblock \bibinfo{publisher}{Springer}.
\newblock


\bibitem[\protect\citeauthoryear{Kock}{Kock}{2011}]%
        {kock}
\bibfield{author}{\bibinfo{person}{Anders Kock}.}
  \bibinfo{year}{2011}\natexlab{}.
\newblock \showarticletitle{Commutative monads as a theory of distributions}.
\newblock \bibinfo{journal}{\emph{Theory and Applications of Categories}}
  \bibinfo{volume}{26} (\bibinfo{date}{Aug.} \bibinfo{year}{2011}).
\newblock


\bibitem[\protect\citeauthoryear{Kozen}{Kozen}{1981}]%
        {kozen-probprog}
\bibfield{author}{\bibinfo{person}{Dexter Kozen}.}
  \bibinfo{year}{1981}\natexlab{}.
\newblock \showarticletitle{Semantics of probabilistic programs}.
\newblock \bibinfo{journal}{\emph{J.~Comput.~Syst.~Sci.}} \bibinfo{volume}{22},
  \bibinfo{number}{3} (\bibinfo{year}{1981}), \bibinfo{pages}{328--350}.
\newblock


\bibitem[\protect\citeauthoryear{Lago and Gavazzo}{Lago and Gavazzo}{2019}]%
        {dal-lago}
\bibfield{author}{\bibinfo{person}{Ugo~Dal Lago} {and}
  \bibinfo{person}{Francesco Gavazzo}.} \bibinfo{year}{2019}\natexlab{}.
\newblock \showarticletitle{On bisimilarity in lambda calculi with continuous
  probabilistic choice}.
\newblock \bibinfo{journal}{\emph{Electron.~Notes~Theoret.~Comput.~Sci.}}
  \bibinfo{volume}{347} (\bibinfo{year}{2019}), \bibinfo{pages}{121 -- 141}.
\newblock
\newblock
\shownote{Proc.~MFPS 2019.}


\bibitem[\protect\citeauthoryear{Laird}{Laird}{2004}]%
        {laird-games}
\bibfield{author}{\bibinfo{person}{James Laird}.}
  \bibinfo{year}{2004}\natexlab{}.
\newblock \showarticletitle{A game semantics of local names and good
  variables}. In \bibinfo{booktitle}{\emph{Proc.~FOSSACS 2004}}.
  \bibinfo{pages}{289--303}.
\newblock


\bibitem[\protect\citeauthoryear{Lambek and Scott}{Lambek and Scott}{1988}]%
        {lambek-scott}
\bibfield{author}{\bibinfo{person}{J Lambek} {and} \bibinfo{person}{P~J
  Scott}.} \bibinfo{year}{1988}\natexlab{}.
\newblock \bibinfo{booktitle}{\emph{Introduction to higher order categorical
  logic}}.
\newblock \bibinfo{publisher}{CUP}.
\newblock


\bibitem[\protect\citeauthoryear{Lew, Cusumano-Towner, Sherman, Carbin, and
  Mansinghka}{Lew et~al\mbox{.}}{2019}]%
        {lew:tracetypes}
\bibfield{author}{\bibinfo{person}{Alexander~K. Lew}, \bibinfo{person}{Marco~F.
  Cusumano-Towner}, \bibinfo{person}{Benjamin Sherman},
  \bibinfo{person}{Michael Carbin}, {and} \bibinfo{person}{Vikash~K.
  Mansinghka}.} \bibinfo{year}{2019}\natexlab{}.
\newblock \showarticletitle{Trace types and denotational semantics for sound
  programmable inference in probabilistic languages}.
\newblock \bibinfo{journal}{\emph{Proc. ACM Program. Lang.}}
  \bibinfo{volume}{4}, \bibinfo{number}{POPL}, Article \bibinfo{articleno}{19}
  (\bibinfo{date}{Dec.} \bibinfo{year}{2019}).
\newblock


\bibitem[\protect\citeauthoryear{Milner}{Milner}{1999}]%
        {milner-pi}
\bibfield{author}{\bibinfo{person}{Robin Milner}.}
  \bibinfo{year}{1999}\natexlab{}.
\newblock \bibinfo{booktitle}{\emph{Communicating and mobile systems - the
  Pi-calculus}}.
\newblock \bibinfo{publisher}{CUP}.
\newblock


\bibitem[\protect\citeauthoryear{Moggi}{Moggi}{1991}]%
        {moggi:computation_and_monads}
\bibfield{author}{\bibinfo{person}{Eugenio Moggi}.}
  \bibinfo{year}{1991}\natexlab{}.
\newblock \showarticletitle{Notions of computation and monads}.
\newblock \bibinfo{journal}{\emph{Inform.~Comput.}} \bibinfo{volume}{93},
  \bibinfo{number}{1} (\bibinfo{year}{1991}), \bibinfo{pages}{55 -- 92}.
\newblock


\bibitem[\protect\citeauthoryear{Murawski and Tzevelekos}{Murawski and
  Tzevelekos}{2016}]%
        {murawski-tzevelekos}
\bibfield{author}{\bibinfo{person}{Andrzej~S. Murawski} {and}
  \bibinfo{person}{Nikos Tzevelekos}.} \bibinfo{year}{2016}\natexlab{}.
\newblock \showarticletitle{Nominal game semantics}.
\newblock \bibinfo{journal}{\emph{Found. Trends Program. Lang.}}
  (\bibinfo{year}{2016}).
\newblock


\bibitem[\protect\citeauthoryear{Murray and Sch{\"o}n}{Murray and
  Sch{\"o}n}{2018}]%
        {birch}
\bibfield{author}{\bibinfo{person}{Lawrence~M. Murray} {and}
  \bibinfo{person}{Thomas~B. Sch{\"o}n}.} \bibinfo{year}{2018}\natexlab{}.
\newblock \showarticletitle{Automated learning with a probabilistic programming
  language: Birch}.
\newblock \bibinfo{journal}{\emph{Annual Reviews in Control}}
  \bibinfo{volume}{46} (\bibinfo{year}{2018}), \bibinfo{pages}{29 -- 43}.
\newblock
\showISSN{1367-5788}


\bibitem[\protect\citeauthoryear{Nori, Hur, Rajamani, and Samuel}{Nori
  et~al\mbox{.}}{2014}]%
        {r2}
\bibfield{author}{\bibinfo{person}{Aditya Nori}, \bibinfo{person}{Chung-Kil
  Hur}, \bibinfo{person}{Sriram Rajamani}, {and} \bibinfo{person}{Selva
  Samuel}.} \bibinfo{year}{2014}\natexlab{}.
\newblock \showarticletitle{{R}2: An efficient {M}{C}{M}{C} sampler for
  probabilistic programs}. In \bibinfo{booktitle}{\emph{Proc.~AAAI 2014}}.
\newblock


\bibitem[\protect\citeauthoryear{Odersky}{Odersky}{1994}]%
        {odersky:localnames}
\bibfield{author}{\bibinfo{person}{Martin Odersky}.}
  \bibinfo{year}{1994}\natexlab{}.
\newblock \showarticletitle{A Functional Theory of Local Names}. In
  \bibinfo{booktitle}{\emph{Proc.~POPL 1994}}. \bibinfo{pages}{48 -- 59}.
\newblock


\bibitem[\protect\citeauthoryear{Orbanz and Roy}{Orbanz and Roy}{2015}]%
        {orbanz-roy}
\bibfield{author}{\bibinfo{person}{Peter Orbanz} {and}
  \bibinfo{person}{Daniel~M. Roy}.} \bibinfo{year}{2015}\natexlab{}.
\newblock \showarticletitle{Bayesian models of graphs, arrays and other
  exchangeable random structures}.
\newblock \bibinfo{journal}{\emph{IEEE Trans. Pattern Anal. Mach. Intell.}}
  \bibinfo{number}{2} (\bibinfo{year}{2015}), \bibinfo{pages}{437--461}.
\newblock


\bibitem[\protect\citeauthoryear{Paquet and Winskel}{Paquet and
  Winskel}{2018}]%
        {paquet-games}
\bibfield{author}{\bibinfo{person}{Hugo Paquet} {and} \bibinfo{person}{Glynn
  Winskel}.} \bibinfo{year}{2018}\natexlab{}.
\newblock \showarticletitle{Continuous probability distributions in concurrent
  games}. In \bibinfo{booktitle}{\emph{Proc.~MFPS 2018}}.
  \bibinfo{pages}{321--344}.
\newblock


\bibitem[\protect\citeauthoryear{Parzygnat}{Parzygnat}{2020}]%
        {quantum-markov}
\bibfield{author}{\bibinfo{person}{Arthur~J. Parzygnat}.}
  \bibinfo{year}{2020}\natexlab{}.
\newblock \bibinfo{title}{Inverses, disintegrations, and Bayesian inversion in
  quantum {M}arkov categories}.
\newblock \bibinfo{howpublished}{arXiv:2001.08375}.
\newblock


\bibitem[\protect\citeauthoryear{Patterson}{Patterson}{2020}]%
        {patterson-thesis}
\bibfield{author}{\bibinfo{person}{Evan Patterson}.}
  \bibinfo{year}{2020}\natexlab{}.
\newblock \emph{\bibinfo{title}{The algebra and machine representation of
  statistical models}}.
\newblock \bibinfo{thesistype}{Ph.D. Dissertation}. \bibinfo{school}{Stanford
  University Department of Statistics}.
\newblock


\bibitem[\protect\citeauthoryear{Pitts}{Pitts}{2013}]%
        {pitts:nominalsets}
\bibfield{author}{\bibinfo{person}{Andrew~M. Pitts}.}
  \bibinfo{year}{2013}\natexlab{}.
\newblock \bibinfo{booktitle}{\emph{Nominal Sets: Names and Symmetry in
  Computer Science}}.
\newblock \bibinfo{publisher}{Cambridge University Press}.
\newblock


\bibitem[\protect\citeauthoryear{Pitts and Stark}{Pitts and Stark}{1993}]%
        {stark:whatsnew}
\bibfield{author}{\bibinfo{person}{Andrew~M. Pitts} {and} \bibinfo{person}{Ian
  Stark}.} \bibinfo{year}{1993}\natexlab{}.
\newblock \showarticletitle{Observable properties of higher order functions
  that dynamically create local names, or: What's {\em new}?}. In
  \bibinfo{booktitle}{\emph{Proc. MFCS 1993}} \emph{(\bibinfo{series}{Lecture
  Notes in Computer Science}, \bibinfo{number}{711})}.
  \bibinfo{pages}{122--141}.
\newblock


\bibitem[\protect\citeauthoryear{Plotkin}{Plotkin}{1973}]%
        {plotkin-lambda-def}
\bibfield{author}{\bibinfo{person}{G.~D. Plotkin}.}
  \bibinfo{year}{1973}\natexlab{}.
\newblock \bibinfo{booktitle}{\emph{Lambda-definability and logical
  relations}}.
\newblock \bibinfo{type}{{T}echnical {R}eport} SAI-RM-4.
  \bibinfo{institution}{School of A.I., Univ.of Edinburgh}.
\newblock


\bibitem[\protect\citeauthoryear{Pollard}{Pollard}{2001}]%
        {pollard}
\bibfield{author}{\bibinfo{person}{David Pollard}.}
  \bibinfo{year}{2001}\natexlab{}.
\newblock \bibinfo{booktitle}{\emph{A users' guide to measure-theoretic
  probability}}.
\newblock \bibinfo{publisher}{CUP}.
\newblock


\bibitem[\protect\citeauthoryear{Roy, Mansinghka, Goodman, and Tenenbaum}{Roy
  et~al\mbox{.}}{2008}]%
        {roy-gensym}
\bibfield{author}{\bibinfo{person}{Daniel Roy}, \bibinfo{person}{Vikash
  Mansinghka}, \bibinfo{person}{Noah Goodman}, {and} \bibinfo{person}{Josh
  Tenenbaum}.} \bibinfo{year}{2008}\natexlab{}.
\newblock \showarticletitle{A stochastic programming perspective on
  nonparametric Bayes}. In \bibinfo{booktitle}{\emph{Proc.~ICML Workshop on
  Nonparametric Bayes}}.
\newblock


\bibitem[\protect\citeauthoryear{Sato, Aguirre, Barthe, Gaboardi, Garg, and
  Hsu}{Sato et~al\mbox{.}}{2019}]%
        {sato:formal-verification}
\bibfield{author}{\bibinfo{person}{Tetsuya Sato}, \bibinfo{person}{Alejandro
  Aguirre}, \bibinfo{person}{Gilles Barthe}, \bibinfo{person}{Marco Gaboardi},
  \bibinfo{person}{Deepak Garg}, {and} \bibinfo{person}{Justin Hsu}.}
  \bibinfo{year}{2019}\natexlab{}.
\newblock \showarticletitle{Formal verification of higher-order probabilistic
  programs: reasoning about approximation, convergence, bayesian inference, and
  optimization}.
\newblock \bibinfo{journal}{\emph{Proc. ACM Program. Lang.}}
  \bibinfo{volume}{3}, \bibinfo{number}{POPL}, Article \bibinfo{articleno}{38}
  (\bibinfo{date}{Jan.} \bibinfo{year}{2019}), \bibinfo{numpages}{30}~pages.
\newblock


\bibitem[\protect\citeauthoryear{{\'S}cibior, Kammar, V{\'a}k{\'a}r, Staton,
  Yang, Cai, Ostermann, Moss, Heunen, and Ghahramani}{{\'S}cibior
  et~al\mbox{.}}{2017}]%
        {scibior}
\bibfield{author}{\bibinfo{person}{Adam {\'S}cibior}, \bibinfo{person}{Ohad
  Kammar}, \bibinfo{person}{Matthijs V{\'a}k{\'a}r}, \bibinfo{person}{Sam
  Staton}, \bibinfo{person}{Hongseok Yang}, \bibinfo{person}{Yufei Cai},
  \bibinfo{person}{Klaus Ostermann}, \bibinfo{person}{Sean Moss},
  \bibinfo{person}{Chris Heunen}, {and} \bibinfo{person}{Zoubin Ghahramani}.}
  \bibinfo{year}{2017}\natexlab{}.
\newblock \showarticletitle{Denotational validation of higher-order Bayesian
  inference}.
\newblock \bibinfo{journal}{\emph{Proceedings of the ACM on Programming
  Languages}}  \bibinfo{volume}{2} (\bibinfo{date}{Nov.} \bibinfo{year}{2017}).
\newblock


\bibitem[\protect\citeauthoryear{Shiebler}{Shiebler}{2020}]%
        {shiebler-likelihood}
\bibfield{author}{\bibinfo{person}{Dan Shiebler}.}
  \bibinfo{year}{2020}\natexlab{}.
\newblock \bibinfo{title}{Categorical stochastic processes and likelihood}.
\newblock \bibinfo{howpublished}{arXiv:2005.04735}.
\newblock


\bibitem[\protect\citeauthoryear{Simpson}{Simpson}{2017}]%
        {simpson-sheaves}
\bibfield{author}{\bibinfo{person}{Alex Simpson}.}
  \bibinfo{year}{2017}\natexlab{}.
\newblock \showarticletitle{Probability Sheaves and the {G}iry Monad}. In
  \bibinfo{booktitle}{\emph{Proc.~CALCO 2017}}.
\newblock


\bibitem[\protect\citeauthoryear{Srivastava}{Srivastava}{1998}]%
        {srivastava}
\bibfield{author}{\bibinfo{person}{Shashi~M. Srivastava}.}
  \bibinfo{year}{1998}\natexlab{}.
\newblock \bibinfo{booktitle}{\emph{A Course on Borel Sets}}.
\newblock \bibinfo{publisher}{Springer, New York}.
\newblock


\bibitem[\protect\citeauthoryear{Stark}{Stark}{1994}]%
        {stark:thesis}
\bibfield{author}{\bibinfo{person}{Ian Stark}.}
  \bibinfo{year}{1994}\natexlab{}.
\newblock \emph{\bibinfo{title}{Names and Higher-Order Functions}}.
\newblock \bibinfo{thesistype}{Ph.D. Dissertation}. \bibinfo{school}{University
  of Cambridge}.
\newblock
\newblock
\shownote{Also available as Technical Report~363, University of Cambridge
  Computer Laboratory.}


\bibitem[\protect\citeauthoryear{Stark}{Stark}{1996}]%
        {stark:cmln}
\bibfield{author}{\bibinfo{person}{Ian Stark}.}
  \bibinfo{year}{1996}\natexlab{}.
\newblock \showarticletitle{Categorical models for local names}.
\newblock \bibinfo{journal}{\emph{{LISP} and Symbolic Computation}}
  \bibinfo{volume}{9}, \bibinfo{number}{1} (\bibinfo{date}{Feb.}
  \bibinfo{year}{1996}), \bibinfo{pages}{77--107}.
\newblock


\bibitem[\protect\citeauthoryear{Staton}{Staton}{2010}]%
        {staton-local-state}
\bibfield{author}{\bibinfo{person}{Sam Staton}.}
  \bibinfo{year}{2010}\natexlab{}.
\newblock \showarticletitle{Completeness for algebraic theories of local
  state}. In \bibinfo{booktitle}{\emph{Proc.~FOSSACS 2010}}.
  \bibinfo{pages}{48--63}.
\newblock


\bibitem[\protect\citeauthoryear{Staton}{Staton}{2017}]%
        {staton:sfinite}
\bibfield{author}{\bibinfo{person}{Sam Staton}.}
  \bibinfo{year}{2017}\natexlab{}.
\newblock \showarticletitle{Commutative semantics for probabilistic
  programming}. In \bibinfo{booktitle}{\emph{Proc.~ESOP 2017}}.
\newblock


\bibitem[\protect\citeauthoryear{Staton, Stein, Yang, Ackerman, Freer, and
  Roy}{Staton et~al\mbox{.}}{2018}]%
        {staton:betabernoulli}
\bibfield{author}{\bibinfo{person}{Sam Staton}, \bibinfo{person}{Dario Stein},
  \bibinfo{person}{Hongseok Yang}, \bibinfo{person}{Nathanael~L. Ackerman},
  \bibinfo{person}{Cameron~E. Freer}, {and} \bibinfo{person}{Daniel~M. Roy}.}
  \bibinfo{year}{2018}\natexlab{}.
\newblock \showarticletitle{The Beta-Bernoulli process and algebraic effects}.
\newblock \bibinfo{journal}{\emph{Proc.~ICALP 2018}}.
\newblock


\bibitem[\protect\citeauthoryear{Staton, Yang, Ackerman, Freer, and Roy}{Staton
  et~al\mbox{.}}{2017}]%
        {syafr-pps}
\bibfield{author}{\bibinfo{person}{S. Staton}, \bibinfo{person}{H. Yang},
  \bibinfo{person}{N.~L.. Ackerman}, \bibinfo{person}{C. Freer}, {and}
  \bibinfo{person}{D. Roy}.} \bibinfo{year}{2017}\natexlab{}.
\newblock \showarticletitle{Exchangeable random process and data abstraction}.
  In \bibinfo{booktitle}{\emph{Proc.~PPS 2017}}.
\newblock


\bibitem[\protect\citeauthoryear{Staton, Yang, Wood, Heunen, and Kammar}{Staton
  et~al\mbox{.}}{2016}]%
        {staton:sheaves}
\bibfield{author}{\bibinfo{person}{Sam Staton}, \bibinfo{person}{Hongseok
  Yang}, \bibinfo{person}{Frank Wood}, \bibinfo{person}{Chris Heunen}, {and}
  \bibinfo{person}{Ohad Kammar}.} \bibinfo{year}{2016}\natexlab{}.
\newblock \showarticletitle{Semantics for probabilistic programming:
  higher-order functions, continuous distributions, and soft constraints}. In
  \bibinfo{booktitle}{\emph{Proc.~LICS 2016}}. \bibinfo{pages}{525 -- 534}.
\newblock


\bibitem[\protect\citeauthoryear{Sumii and Pierce}{Sumii and Pierce}{2003}]%
        {sumii-pierce}
\bibfield{author}{\bibinfo{person}{Eijiro Sumii} {and}
  \bibinfo{person}{Benjamin~C. Pierce}.} \bibinfo{year}{2003}\natexlab{}.
\newblock \showarticletitle{Logical relations for encryption}.
\newblock \bibinfo{journal}{\emph{J. Comput. Secur.}} \bibinfo{volume}{11},
  \bibinfo{number}{4} (\bibinfo{year}{2003}), \bibinfo{pages}{521--554}.
\newblock


\bibitem[\protect\citeauthoryear{Tzevelekos}{Tzevelekos}{2008}]%
        {tzevelekos-thesis}
\bibfield{author}{\bibinfo{person}{Nikos Tzevelekos}.}
  \bibinfo{year}{2008}\natexlab{}.
\newblock \emph{\bibinfo{title}{Nominal game semantics}}.
\newblock \bibinfo{thesistype}{Ph.D. Dissertation}. \bibinfo{school}{Oxford
  University Computing Laboratory}.
\newblock


\bibitem[\protect\citeauthoryear{van~de Meent, Paige, Yang, and Wood}{van~de
  Meent et~al\mbox{.}}{2018}]%
        {mpyw-probprog-intro}
\bibfield{author}{\bibinfo{person}{Jan-Willem van~de Meent},
  \bibinfo{person}{Brooks Paige}, \bibinfo{person}{Hongseok Yang}, {and}
  \bibinfo{person}{Frank Wood}.} \bibinfo{year}{2018}\natexlab{}.
\newblock \bibinfo{title}{An introduction to probabilistic programming}.
\newblock \bibinfo{howpublished}{arxiv:1809.10756}.
\newblock


\bibitem[\protect\citeauthoryear{Vandenbroucke and Schrijvers}{Vandenbroucke
  and Schrijvers}{2020}]%
        {plonk}
\bibfield{author}{\bibinfo{person}{Alexander Vandenbroucke} {and}
  \bibinfo{person}{Tom Schrijvers}.} \bibinfo{year}{2020}\natexlab{}.
\newblock \showarticletitle{P$\lambda\omega${N}{K}: functional probabilistic
  {N}et{K}{A}{T}}. In \bibinfo{booktitle}{\emph{Proc.~POPL 2020}}.
\newblock


\bibitem[\protect\citeauthoryear{Zhang and Nowak}{Zhang and Nowak}{2003}]%
        {zhang-nowak}
\bibfield{author}{\bibinfo{person}{Yu Zhang} {and} \bibinfo{person}{David
  Nowak}.} \bibinfo{year}{2003}\natexlab{}.
\newblock \showarticletitle{Logical relations for dynamic name creation}. In
  \bibinfo{booktitle}{\emph{Proc.~CSL 2003}}. \bibinfo{pages}{575--588}.
\newblock


\end{thebibliography}

\end{document}